\documentclass[journal,10pt,twoside,twocolumn]{IEEEtran}

\usepackage{cite}
\usepackage[T1]{fontenc}
\usepackage{graphicx}
\usepackage{amssymb}
\usepackage{amsmath}
\usepackage{amsthm}
\usepackage{paralist}
\usepackage{setspace}
\usepackage{microtype}
\usepackage{hyperref}
\usepackage{url}
\usepackage{balance}
\usepackage{subfigure}
\usepackage{xcolor}
\usepackage{fancyref}
\usepackage{framed}
\usepackage{algorithm}
\usepackage{nicefrac}
\usepackage{cases}
\usepackage{gensymb}
\usepackage{booktabs}
\usepackage{multirow}
\usepackage{graphicx}

\newtheorem{theorem}{Theorem}
\newtheorem{corollary}{Corollary}
\newtheorem{lemma}{Lemma}
\newtheorem{definition}{Definition}
\newtheorem{remark}{Remark}

\allowdisplaybreaks 

\newcommand{\cshere}[1]{\begin{center}\textcolor{red}{\bf--- CS stopped his edits here ---}\end{center}}
\newcommand{\reals}{ \mathbb{R} }
\newcommand{\la}{\langle }
\newcommand{\ra}{\rangle}
\newcommand{\comps}{\mathbb{C}}

\DeclareMathOperator*{\ang}{angle}
\DeclareMathOperator*{\phase}{phase}
\newcommand{\rsphere}[1]{\setS^{{#1}-1}_\Real}
\newcommand{\csphere}[1]{\setS^{{#1}-1}_\comps}
\newcommand{\hsphere}[1]{\setS^{{#1}-1}_\setH}

\renewcommand{\eqref}[1]{(\ref{#1})}
\newcommand{\secref}[1]{\mbox{Section~\ref{#1}}}
\newcommand{\figref}[1]{\mbox{Figure~\ref{#1}}}
\newcommand{\tblref}[1]{\mbox{Table~\ref{#1}}}
\newcommand{\thmref}[1]{\mbox{Theorem~\ref{#1}}}
\newcommand{\lemref}[1]{\mbox{Lemma~\ref{#1}}}
\newcommand{\corref}[1]{\mbox{Corollary~\ref{#1}}}

\renewcommand{\mid}{\ensuremath{\,|\,}}

\newcommand{\diag}[1]{\ensuremath{\mathrm{diag}\!\left(#1\right)}}

\newcommand{\Real}{\ensuremath{\mathbb{R}}}
\newcommand{\Complex}{\ensuremath{\mathbb{C}}}

\newcommand{\setA}{\ensuremath{\mathcal{A}}}
\newcommand{\setB}{\ensuremath{\mathcal{B}}}
\newcommand{\setC}{\ensuremath{\mathcal{C}}}
\newcommand{\setD}{\ensuremath{\mathcal{D}}}

\newcommand{\setH}{\ensuremath{\mathcal{H}}}

\newcommand{\setP}{\ensuremath{\mathcal{P}}}

\newcommand{\setS}{\ensuremath{\mathcal{S}}}

\newcommand{\bma}{\ensuremath{\mathbf{a}}}
\newcommand{\bmb}{\ensuremath{\mathbf{b}}}

\newcommand{\bmx}{\ensuremath{\mathbf{x}}}
\newcommand{\bmy}{\ensuremath{\mathbf{y}}}
\newcommand{\bmz}{\ensuremath{\mathbf{z}}}

\newcommand{\bmdelta}{\ensuremath{\boldsymbol\delta}}

\newcommand{\bA}{\ensuremath{\mathbf{A}}}
\newcommand{\bB}{\ensuremath{\mathbf{B}}}

\newcommand{\bDelta}{\ensuremath{\mathbf{\Delta}}}

\newcommand{\half}{\frac{1}{2}}

\newcommand{\alns}[1]{\begin{align*}#1\end{align*}}
\newcommand{\aln}[1]{\begin{align}#1\end{align}}

\DeclareMathOperator*{\st}{subject\,\,to}
\newcommand{\matb}{\left( \begin{matrix*}[r] }
\newcommand{\mate}{\end{matrix*}\right)}

\newcommand{\opt}{^\star}

\DeclareMathOperator*{\minimize}{minimize}
\DeclareMathOperator*{\maximize}{maximize}

\usepackage{color}
\newcommand{\revised}[1]{#1} 

\begin{document}

\title{\scalebox{0.89}{PhaseMax: Convex Phase Retrieval via Basis Pursuit}}\author{Tom Goldstein and Christoph Studer\thanks{T. Goldstein is with the Department of Computer Science, University of Maryland, College Park, MD (e-mail: \url{ tomg@cs.umd.edu}).}\thanks{C.~Studer is with the School of Electrical and Computer Engineering, Cornell University, Ithaca, NY (e-mail: \url{studer@cornell.edu}).}
\thanks{This paper was presented in part at the 32th International Conference on Machine Learning (ICML) \cite{pmlr-v70-goldstein17a}.}
\thanks{The work of T.~Goldstein was supported in part by the US National Science Foundation (NSF) under grant CCF-1535902, the US Office of Naval Research under grant N00014-17-1-2078, and by the Sloan Foundation. The work of C. Studer was supported in part by Xilinx, Inc. and by the US NSF under grants ECCS-1408006, CCF-1535897,  CNS-1717559, and CAREER CCF-1652065.}}

\maketitle

\begin{abstract}
We consider the recovery of a (real- or complex-valued) signal from magnitude-only measurements, known as phase retrieval. We formulate phase retrieval as a convex optimization problem, which we call PhaseMax. Unlike other convex methods that use semidefinite relaxation and lift the phase retrieval problem to a higher dimension, PhaseMax is a ``non-lifting'' relaxation that operates in the original signal dimension. We show that the dual problem to PhaseMax is Basis Pursuit, which implies that phase retrieval can be performed using algorithms initially designed for sparse signal recovery. We develop sharp lower bounds on the success probability of PhaseMax for a broad range of random measurement ensembles, and we analyze the impact of measurement noise on the solution accuracy. We use numerical results to demonstrate the accuracy of our recovery guarantees, and we showcase the efficacy and limits of PhaseMax in practice. 
\end{abstract}



\section{Introduction}
\label{sec:introduction}
Phase retrieval is concerned with the recovery of an $n$-dimensional signal $\bmx^0\in\setH^n$, with $\setH$ either~$\Real$ or~$\Complex$, from $m\geq n$ squared-magnitude, noisy measurements~\cite{candes2013phaselift} 
\aln{ \label{eq:original} 
b_i^2 = |\la\bma_i, \bmx^0\ra|^2 + \eta_i, \quad i = 1,2,\ldots,m,
}
where $\bma_i\in\setH^n$, $i=1,2,\ldots,m$, are the (known) measurement vectors and $\eta_i\in\Real$, $i=1,2,\ldots,m$, models measurement noise. 
Let $\hat\bmx\in\setH^n$ be an approximation vector\footnote{Approximation vectors can be obtained via a variety of algorithms, or can even be chosen at random. See \secref{sec:approx}.} to the true signal $\bmx^0$. We recover the signal~$\bmx^0$ by solving the following convex problem we call {\em PhaseMax}:  
\alns{   
\text{(PM)} \quad \left\{\begin{array}{ll} 
\underset{\bmx\in\setH^n}{\maximize} &  \langle \bmx,\hat \bmx \rangle_\Re     \\
\st   & |\la\bma_i,\bmx\ra| \le b_i, \quad i=1,2,\ldots,m.
\end{array}\right.
}
Here, $\langle \bmx,\hat \bmx \rangle_\Re$ denotes the real-part of the inner product between the vectors $\bmx$ and $\hat\bmx$. The main idea behind PhaseMax is to find the vector $\bmx$ that is most aligned with the approximation vector $\hat\bmx$ and satisfies a convex relaxation of the measurement constraints in \eqref{eq:original}. 

Our main goal is to develop sharp lower bounds on the probability with which PhaseMax succeeds in recovering the true signal~$\bmx^0$, up to an arbitrary phase ambiguity that does not affect the measurement constraints in~\eqref{eq:original}.
By assuming noiseless measurements, one of our main results is as follows.
\begin{theorem} \label{thm:reconC}
Consider the case of recovering a complex-valued signal $\bmx\in\Complex^n$ from $m$ noiseless measurements of the form \eqref{eq:original} with  measurement vectors $\bma_i$, $i=1,2,\ldots,m$, sampled independently and uniformly from the unit sphere.  Let  
\begin{align*}
\ang( \bmx^0,\hat\bmx) = \arccos\!\left( \frac{\langle \bmx^0, \hat\bmx\rangle_\Re}{\|\bmx^0\|_2\|\hat\bmx\|_2} \right)
\end{align*}
be the angle between the true vector $\bmx^0$ and the approximation~$\hat\bmx$, and define the constant 
\begin{align*}
\alpha= 1-\textstyle\frac{2}{\pi}\ang( \bmx^0,\hat\bmx)
\end{align*}
that measures the approximation accuracy. Then, the probability  that PhaseMax recovers the true signal~$\bmx^0$, denoted by $p_\Complex(m,n)$, is bounded from below as follows:  
\aln{ \label{eq:pcomplex}
p_\Complex(m,n) \ge 1-\exp\!\left(\! - \frac{\left(\alpha m -4n\right)^2}{2m} \right) 
}
whenever $\alpha m >4n.$
\end{theorem}
In words, if $m> 4n/\!\alpha$ and $\alpha>0$, then PhaseMax will succeed with non-zero probability. Furthermore, for a fixed signal dimension $n$ and an arbitrary approximation vector $\hat\bmx$ that satisfies $\ang(\bmx^0,\hat\bmx)<\frac{\pi}{2}$, i.e., one that is not orthogonal to the vector $\bmx^0$, we can make the success probability of PhaseMax arbitrarily close to one by increasing the number of measurements~$m$. As we shall see, our recovery guarantees  are sharp and accurately predict the performance of PhaseMax in practice.  

\revised{We emphasize that the convex formulation~(PM) has been studied by several other authors, including Bahmani and Romberg \cite{bahmani2017phase}, whose work appeared shortly before our own. Related work will be discussed in detail in Section \ref{sec:related}.}

\subsection{Convex Phase Retrieval via Basis Pursuit}
It is quite intriguing that the following Basis Pursuit problem \cite{chen2001atomic,chen1994basis}
\alns{   
\text{(BP)} \quad \left\{\begin{array}{ll} 
\underset{\bmz\in\setH^m}{\minimize} &  \|\bmz\|_1     \\
\st   &\hat\bmx=\bA \bB^{-1}\bmz,
\end{array}\right.
}
with $\bB=\diag{b_1,b_2,\ldots,b_m}$ and $\bA=[\bma_1,\bma_2,\ldots,\bma_m]$ is the dual problem to (PM); see, e.g., \cite[Lem.~1]{studer2012signal}.
As a consequence, if PhaseMax succeeds, then the phases of the solution vector $\bmz\in\setH^m$ to (BP) are exactly the phases that were lost in the measurement process in~\eqref{eq:original}, i.e., we have 
\begin{align*}
y_i=\phase(z_i)b_i=\la\bma_i,\bmx^0\ra, \quad i=1,2,\ldots,m,
\end{align*}
with $\phase(z)=z/|z|$ for $z\neq0$ and $\phase(0)=1.$  This observation not only reveals a fundamental connection between phase retrieval and sparse signal recovery, but also implies that Basis Pursuit solvers can be used to recover the signal from the phase-less measurements in \eqref{eq:original}.

\subsection{\revised{A Brief History of Phase Retrieval}}
\label{sec:priorart}
Phase retrieval is a well-studied problem with a long history \cite{gerchberg1972practical,fienup1982phase} and enjoys widespread use in applications such as X-ray crystallography \cite{harrison1993phase,miao2008extending,pfeiffer2006phase}, microscopy \cite{kou2010transport,faulkner2004movable}, imaging~\cite{holloway2016}, and many more~\cite{fogel2013phase,candes2015phasea,jaganathan2015phase,tian20153d}. 
Early algorithms, such as the  Gerchberg-Saxton \cite{gerchberg1972practical} or Fienup \cite{fienup1982phase} algorithms, rely on alternating projection to recover complex-valued signals from magnitude-only measurements.   
The papers~\cite{candes2013phaselift,candes2014solving,candes2015phase} sparked new interest in the phase retrieval problem by showing that  it can be relaxed to a semidefinite program. Prominent convex relaxations include PhaseLift \cite{candes2013phaselift} and PhaseCut \cite{waldspurger2015phase}.  These methods come with recovery guarantees, but require the problem to be lifted to a higher dimensional space, which prevents their use for large-scale problems.
\revised{Several authors recently addressed this problem by presenting efficient methods that can solve large-scale semidefinite programs without the complexity of representing the full-scale (or lifted) matrix of unknowns. These approaches include gauge duality methods \cite{friedlander2014gauge,aravkin2017foundations} and sketching methods~\cite{yurtsever2017sketchy}.}

\revised{More recently, a number of \emph{non-convex} algorithms have been proposed (see e.g.,  \cite{netrapalli2013phase,schniter2015compressive,candes2015wirtinger,chen2015solving,wang2016solving,wei2015solving,zeng2017coordinate}) that directly operate in the original signal dimension and exhibit excellent empirical performance.
The algorithms in~\cite{netrapalli2013phase,candes2015wirtinger,chen2015solving,wang2016solving,yuan2017phase} come with recovery guarantees that mainly rely on accurate initializers, such as the (truncated) spectral initializer \cite{netrapalli2013phase,chen2015solving}, the Null initializer~\cite{chen2015phase}, the orthogonality-promoting method~\cite{wang2016solving}, or more recent initializers that guarantee optimal performance \cite{lu2017phase,mondelli2017fundamental} (see \secref{sec:approx} for additional details). These initializers enable non-convex phase retrieval algorithms to succeed, given a sufficiently large number of measurements; see~\cite{sun2016geometric} for more details on the geometry of such non-convex problems. 
}

\subsection{\revised{Related Work on Non-Lifting Convex Methods}}
\label{sec:related}
\revised{Because of the extreme computational burden of traditional convex relaxations (which require ``lifting'' to a higher dimension), a number of authors have recently been interested in non-lifting convex relaxations for phase retrieval.  Shortly before the appearance of this work, the formulation (PM) was studied by Bahmani and Romberg~\cite{bahmani2017phase}.  Using methods from machine learning theory, the authors derived bounds on the recovery of signals with and without noise.  The results in \cite{bahmani2017phase} are stronger than those presented here in that they are uniform with respect to the approximation vector $\hat \bmx,$ but weaker in the sense that they require significantly more measurements.  
}

\revised{Several authors have studied PhaseMax after the initial appearance of this work.   An alternative proof of accurate signal recovery was derived using measure concentration bounds in \cite{hand2016elementary}.  Compressed sensing methods for sparse signals \cite{hand2016compressed}, and corruption-robust methods for noisy signals \cite{hand2016corruption} have also been presented.  The closely related non-lifting relaxation BranchHull has also been proposed for recovering signals from entry-wise product measurements~\cite{aghasi2017branchhull}.
}

\revised{Finally, we note that recent works have proved tight asymptotic bounds for PhaseMax in the asymptotic limit, i.e., where $\beta=m/n$ is constant and $m\to\infty$. The authors of \cite{dhifallah2017phase} derive an exact asymptotic bound of the performance of PhaseMax, and a related non-rigorous analysis was given in \cite{dhifallah2017fundamental}.  It was also shown that better signal recovery guarantees can be obtained by iteratively applying PhaseMax.  The resulting method, called PhaseLamp, was analyzed in \cite{dhifallah2017phase}.} 



\subsection{Contributions and Paper Outline}
In contrast to algorithms relying on semidefinite relaxation or non-convex problem formulations, we propose \emph{PhaseMax}, a novel, convex method for phase retrieval that directly operates in the original signal dimension.
In \secref{sec:optimality}, we establish a deterministic condition that guarantees uniqueness of the solution to the (PM) problem.
We borrow methods from geometric probability to derive sharp lower bounds on the success probability for real- and complex-valued systems in \secref{sec:preliminaries}. \secref{sec:generalizations} generalizes our results to a broader range of random measurement ensembles and to systems with measurement noise.
We show in \secref{sec:approx} that randomly chosen approximation vectors are sufficient to ensure faithful recovery, given a sufficiently large number of measurements.
We numerically demonstrate the sharpness of our recovery guarantees and showcase the practical limits of PhaseMax in \secref{sec:discussion}.
We conclude in \secref{sec:conclusions}.

\subsection{Notation}
Lowercase and uppercase boldface letters stand for column vectors and matrices, respectively. For a complex-valued matrix~$\bA$, we denote its transpose and  Hermitian transpose by $\bA^T$ and $\bA^*$, respectively; the real and imaginary parts are $\bA_\Re$ and $\bA_\Im$.
The $i$th column of the matrix $\bA$ is denoted by~$\bma_i$ and the $k$th entry of the $i$th vector $\bma_i$ is $[\bma_i]_k$; for a vector~$\bma$ without index, we simply denote the $k$th entry by~$a_k$.  We define the inner product between two complex-valued vectors~$\bma$ and~$\bmb$ as $\langle\bma,\bmb\rangle=\bma^*\bmb$. We use $j$ to denote the imaginary unit.
The $\ell_2$-norm and $\ell_1$-norm of the vector~$\bma$ are $\|\bma\|_2$ and $\|\bma\|_1$, respectively. 


\section{Uniqueness Condition}
\label{sec:optimality}

The measurement constraints in \eqref{eq:original} do not uniquely define a vector. If  $\bmx$ is a vector that satisfies \eqref{eq:original}, then any vector $\bmx'=e^{j\phi}\bmx$ for $\phi\in[0,2\pi)$ also satisfies the constraints.
In contrast, if~$\bmx$ is a solution to (PM), then $e^{j\phi}\bmx$  with $\phi\neq0$ will {\em not} be another solution.
In fact, consider any vector~$\bmx$ in the feasible set of (PM) with $\langle  \bmx,\hat \bmx\rangle_\Im \neq 0$.  By choosing $\omega = \phase(\langle \bmx,\hat\bmx\rangle)$, we have 
$$\la\omega\bmx,\hat\bmx \ra_\Re  = |\la \bmx,\hat\bmx\ra| > \la \bmx ,\hat\bmx\ra_\Re,$$
which implies that given such a vector $\bmx,$ one can always increase the objective function of (PM) simply by {\em aligning} $\bmx$ with the approximation $\hat \bmx$ (i.e., modifying its phase so that $\la\bmx,\hat\bmx\ra$ is real valued).  The following definition makes this observation rigorous.
\begin{definition}
A vector $\bmx$ is said to be {\em aligned} with another vector $\hat \bmx$, if the inner product $\la\bmx,\hat\bmx\ra$ is real-valued and non-negative.
\end{definition}
From all the vectors that satisfy the measurement constraints in \eqref{eq:original}, there is only one that is a candidate solution to the convex  problem (PM), which is also the solution that is aligned with $\hat \bmx.$  For this reason, we adopt the following important convention throughout the rest of this paper.  
\begin{framed}
\centering
The true vector $\bmx^0$ denotes a solution to \eqref{eq:original} that is aligned with the approximation vector $\hat \bmx.$
\end{framed}
\begin{remark}
There is an interesting relation between the convex formulation of PhaseMax and the semidefinite relaxation method PhaseLift~\cite{candes2013phaselift,candes2014solving,candes2015phase}. Recall that the set of solutions to any convex problem is always convex.  However, the solution set of the measurement constraints \eqref{eq:original} is invariant under phase rotations, and thus non-convex (provided it is non-zero). It is therefore impossible to design a convex problem that yields this set of solutions.  PhaseMax and PhaseLift differ in how they remove the phase ambiguity from the problem to enable a convex formulation.  Rather than trying to identify the true vector $\bmx^0,$ PhaseLift reformulates the problem in terms of the quantity $\bmx^0(\bmx^0)^H,$ which is unaffected by phase rotations in $\bmx^0.$  Hence, PhaseLift removes the rotation symmetry from the solution set, yielding a problem with a convex set of solutions.  PhaseMax does something much simpler: it {\em pins down} the phase of the solution to an arbitrary quantity, thus removing the phase ambiguity and restoring convexity to the solution set.  This arbitrary phase choice is made when selecting the phase of the approximation $\hat \bmx.$
\end{remark}

We are now ready to state a deterministic condition under which PhaseMax succeeds in recovering the true vector~$\bmx^0.$  The result applies to the noiseless case, i.e., $\eta_i=0$, $i=1,2,\ldots,m$. In this case, all inequality constraints in (PM) are active at $\bmx^0.$ The noisy case will be discussed in \secref{sec:noise}.

\begin{theorem} \label{thm:opt}
The true vector $\bmx^0$ is the unique maximizer of (PM) if, for any unit vector $\bmdelta\in\setH^n$ that is aligned with the approximation $\hat \bmx,$
 \alns{
    \exists i,   [ \la \bma_i,\bmx^0 \ra^* \la \bma_i, \bmdelta\ra ]_\Re > 0.
 }
\end{theorem}
\begin{proof}

Suppose the conditions of this theorem hold, and consider some candidate solution $\bmx'$ in the feasible set for (PM) with $\langle \bmx',\hat \bmx \rangle \ge  \langle \bmx^0,\hat \bmx \rangle.$ Without loss of generality, we assume~$\bmx'$ to be aligned with~$\hat \bmx.$ Then, the vector $\bDelta = \bmx'-\bmx^0$ is also aligned with $\hat \bmx$, and satisfies 
$$\langle \bDelta,\hat \bmx \rangle = \langle \bmx',\hat \bmx \rangle- \langle \bmx^0,\hat \bmx \rangle \ge 0.$$
Since $\bmx'$ is a feasible solution for (PM), we have
\begin{align*}
& |\la \bma_i, \bmx^0+\bDelta \ra|^2 =  \\
& \quad |\la \bma_i, \bmx^0 \ra|^2+2 [ \la \bma_i, \bmx^0 \ra^* \la \bma_i, \bDelta \ra]_\Re+ |\la\bma_i,\bDelta\ra|^2 \le b_i^2, \,\, \forall i .
\end{align*}
But  $|\bma_i^T\bmx^0|^2=b_i^2,$ and so
$$ [ \la \bma_i, \bmx^0 \ra^* \la \bma_i, \bDelta \ra]_\Re \le - \textstyle \frac{1}{2} |\la\bma_i, \bDelta\ra|^2 \le 0  , \,\, \forall i.  $$
Now, if $\|\bDelta\|_2>0,$ then the unit-length vector $\bmdelta=\bDelta/\|\bDelta\|_2$ satisfies $ [ \la \bma_i, \bmx^0 \ra^* \la \bma_i, \bmdelta \ra]_\Re \le 0$ for all $i,$ which contradicts the hypothesis of the theorem. It follows that $\|\bDelta\|_2=0$ and $\bmx'=\bmx^0.$
\end{proof}

Theorem \ref{thm:opt} has an intuitive geometrical interpretation.  If $\bmx^0$ is an optimal point and $\bmdelta$ is an ascent direction, then one cannot move in the direction of $\bmdelta$ starting at  $\bmx^0$ without leaving the feasible set.  This condition is met if there is an  $\bma_i$ such that $\bmx^0$ and $\bmdelta$ both lie on the same side of the plane through the origin orthogonal to the measurement vector $\bma_i.$

\section{Preliminaries:  classical sphere covering problems  and geometric probability}
\label{sec:preliminaries}

To derive sharp conditions on the success probability of PhaseMax, we require a set of  tools from geometric probability.
Many classical problems in geometric probability involve calculating the likelihood of a sphere being covered by random ``caps,'' or semi-spheres, which we define below.

\begin{definition} \label{def:caps}
Consider the set $\setS^{n-1}_{\setH}=\{\bmx\in \setH^n\,|\, \|\bmx\|_2=1\},$ the unit sphere embedded in $\setH^n.$  Given a vector $\bma \in \setH^n,$ the cap centered at $\bma$ with central angle $\theta$ is defined as
 \aln{ \label{smallcap}
  \setC_\setH(\bma,\theta) = \{\bmdelta \in \setS^{n-1}_\setH | \, \la \bma,\bmdelta\ra_\Re > \cos(\theta)\}.
  }
  This cap contains all vectors that form an angle with $\bma$ of less than $\theta$ radians.  When $\theta=\pi/2,$ we have a semisphere centered at $\bma,$ which is simply denoted by
  \aln{ \label{cap}
  \setC_\setH(\bma)= \setC_\setH(\bma,\pi/2) = \{\bmdelta \in \setS^{n-1}_\setH | \, \la \bma,\bmdelta\ra_\Re > 0\}.
  }

\end{definition}

We say that a collection of caps {\em covers} the entire sphere if the sphere is contained in the union of the caps.  Before we can say anything useful about when a collection of caps covers the sphere, we will need the following classical result, which is often attributed to Schl\"afli~\cite{schlaefli1953}. 
Proofs that use simple induction methods can be found in \cite{wendel1962problem,gilbert1965probability,furedi1986random}.  For completeness, we briefly include a short proof in Appendix \ref{sec:cutproof}.
\begin{lemma} \label{realcut}
Consider a sphere $\rsphere{n}\subset \reals^n.$  Suppose we slice the sphere with $k$ planes through the origin.  Suppose also that every subset of $n$ planes have linearly independent normal vectors.  These planes divide the sphere into 
  \aln{ \label{numslices}
  r(n,k)=2\sum_{i=0}^{n-1}{ k-1 \choose i}
  }
regions.
\end{lemma}

Classical results in geometric probability study the likelihood of a sphere being covered by random caps with centers chosen  independently and uniformly from the sphere's surface.  For our purposes, we need to study the more specific case in which caps are only chosen from a subset of the sphere.  While calculating this probability is hard in general, it is quite simple when the set obeys the following symmetry condition.

\begin{definition} \label{symmetric}
We say that the set $\setA$ is {\em symmetric} if, for all $\bmx \in \setA,$ we also have $-\bmx \in \setA.$  
\end{definition}

We are now ready to prove a general result that states when the sphere is covered by random caps. 

\begin{lemma} \label{lem:coins}
Consider some symmetric set $\setA \subset \setS^{n-1}_{\reals}$ of positive ($n-1$ dimensional) measure.  Choose some set of $m_\setA$ measurements~$\{\bma_i\}_{i=1}^{m_\setA}$ uniformly from $\setA.$  Then, the caps $\{\setC_\Real(\bma_i)\}$ cover the sphere $\rsphere{n}$ with probability 
 $$ p_\text{cover}(m_\setA,n) = 1-\frac{1}{2^{m_\setA-1}}   \sum_{k=0}^{n-1} {{m_\setA-1}\choose{k}}.$$
 This is the probability of turning up $n$ or more heads when flipping $m_\setA-1$ fair coins.
\end{lemma}
\begin{proof}
Consider the following two-step process for constructing the set $\{\bma_i\}.$  First, we sample $m_\setA$ vectors~$\{\bma'_i\}$ independently and uniformly from $\setA.$ Note that, because $\setA$ has positive measure, it holds with probability 1 that any subset of $n$ vectors will be linearly independent.

Second, we define $\bma_i = c_i \bma'_i,$ where $\{c_i\}$ are i.i.d. Bernoulli variables that take value $+1$ or $-1$ with probability $\nicefrac{1}{2}.$ 
 We can think of this second step as randomly ``flipping''  a subset of uniform random vectors.  Since $\setA$ is symmetric and $\{\bma'_i\}$ is sampled independently and uniformly, the random vectors $\{\bma_i\}$ also have an independent and uniform distribution over~$\setA.$ This construction may seem superfluous since both~$\{\bma_i\}$ and~$\{\bma'_i\}$ have the same distribution, but we will see below that this becomes useful.

 Given a particular set of coin flips $\{c_i\},$ we can write the set of points that are {\em not} covered by the caps~$\{\setC_\Real(\bma_i)\}$ as 
   \aln{ \label{capintersect}
   \bigcap_i \setC_\Real(-\bma_i) = \bigcap_i  \setC_\Real(-c_i\bma'_i).
   }
 Note that there are $2^{m_\setA}$ such intersections that can be formed, one for each choice of the sequence $\{c_i\}$.  The caps $\{\setC_\Real(\bma_i)\}$ cover the sphere whenever the intersection \eqref{capintersect} is empty.  
 Consider the set of planes $\big\{\{\bmx \mid  \la \bma_i, \bmx\ra =0\}\big\}.$ From Lemma~\ref{realcut}, we know that $m_\setA$ planes with a common intersection point divide the sphere into  
 $ r(n,m_\setA)$
 non-empty regions.    Each of these regions corresponds to the intersection \eqref{capintersect} for one possible choice of $\{c_i\}.$
Therefore, of the $2^{m_\setA}$ possible intersections, only $r(n,m_\setA)$ of them are non-empty.  Since the sequence $\{c_i\}$ is random, each intersection is equally likely to be chosen, and so the probability of covering the sphere is   
$$p_\text{cover}(m_\setA,n) = 1-\frac{r(n,m_\setA)}{2^{m_\setA}}.$$
\end{proof}

\begin{remark}
Several papers have studied the probability of covering the sphere using points independently and uniformly chosen over the entire sphere. The only aspect that is unusual about Lemma \ref{lem:coins} is the observation that this probability remains the same if we restrict our choices to the set $\setA,$ provided $\setA$ is symmetric.  We note that this result was observed by Gilbert~\cite{gilbert1965probability} in the case $n=3,$ and we generalize it to any $n>1$ using a similar argument.
\end{remark}

We now present a somewhat more complicated covering theorem.  The next result considers the case where the measurement vectors are drawn only from a semisphere.  We consider the question of whether these vectors cover enough area to contain not only their home semisphere, but another nearby semisphere as well.

\begin{figure}
\centering
\includegraphics[width=0.97\columnwidth]{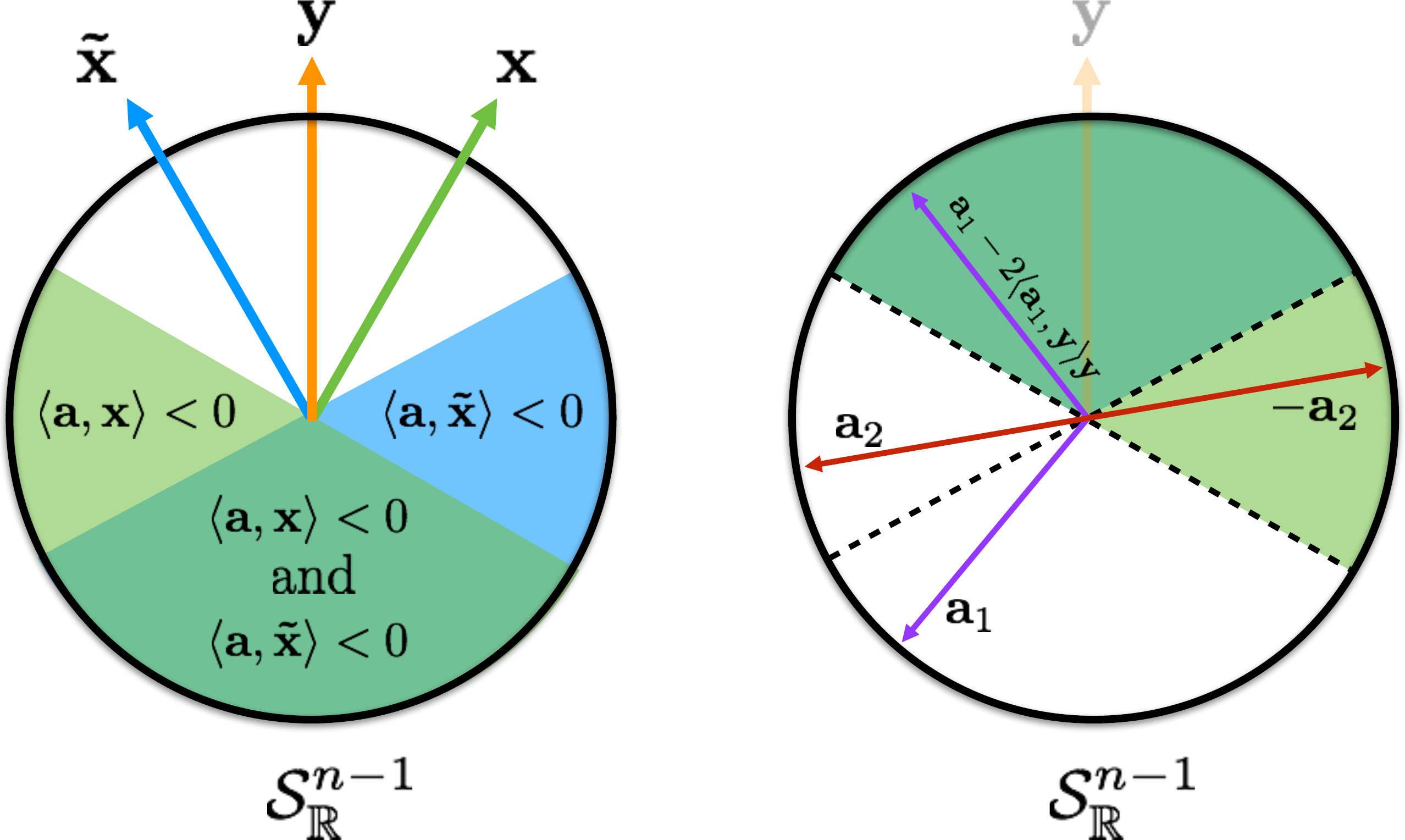}
\caption{(left) A diagram showing the construction of $\bmx,$ $\bmy,$ and $\tilde\bmx$ in the proof of Lemma \ref{lem:neighbor}. (right) The reflections defined in (\ref{case1}-\ref{case3}) map vectors $\bma_i$ lying in the half-sphere  $\bma_i^T \bmx <0$ onto the half-sphere $\bma^T \bmx >0.$}
\label{fig:reflect}
\end{figure}

\begin{lemma} \label{lem:neighbor}
Consider two vectors $\bmx,\bmy \subset \setS^{n-1}_{\reals},$ and the caps $\setC_\Real(\bmx)$ and $\setC_\Real(\bmy).$ Let $\alpha= 1-\frac{2}{\pi}\ang(\bmx,\bmy)$ be a measure of the similarity between the vectors $\bmx$ and $\bmy.$ Draw some collection $\{\bma_i \in \setC_\Real(\bmx)\}_{i=1}^m$ of~$m$ vectors uniformly from $\setC_\Real(\bmx)$ so that $m>2n/\alpha .$
  Then 
   $$\setC_\Real(\bmy) \subset \bigcup_i \setC_\Real(\bma_i)$$  
with probability at least
$$p_{\text{cover}}(m,n;\bmx,\bmy)  \ge 
1 - \exp \!\left (-  \frac{ (\alpha m-2n )^2 }{2m} \right)\!.$$
\end{lemma}
\begin{proof}
Due to rotational symmetry,  we assume $\bmy=[1,0,\ldots,0]^T$ without loss of generality.  
Consider the point $\tilde \bmx = [x_1,-x_2, \ldots,-x_n]^T.$   This is the reflection of $\bmx$ over $\bmy$ (see Figure \ref{fig:reflect}).   
Suppose we have some collection~$\{\bma_i\}$ independently and uniformly distributed on the entire sphere.  Consider the collection of vectors 
 \begin{numcases} {\bma'_i =}
  \bma_i ,  & if  $\la \bma_i,\bmx\ra   \ge 0$ \label{case1}\\
  \bma_i - 2\la \bma_i, \bmy\ra \bmy  &  if  $\la \bma_i,\bmx\ra  < 0,   \la\bma_i, \tilde \bmx\ra  < 0$ \label{case2}\\
   -\bma_i & if  $\la \bma_i,\bmx\ra < 0,   \la \bma_i, \tilde \bmx\ra  \ge 0.$ \label{case3}
\end{numcases}
The mapping $\bma_i \to \bma'_i$ maps the lower half sphere $\{\bma\, |\,  \la \bma,\bmx\ra  <0\}$ onto the upper half sphere $\{\bma \,|\, \la\bma,\bmx\ra >0\}$ using a combination of reflections and translations (depicted in Figure~\ref{fig:reflect}).   Indeed, for all $i$ we have  $\la\bma'_i,\bmx\ra \ge 0.$  This is clearly true in case \eqref{case1} and  \eqref{case3}.  In case \eqref{case2}, we use the definition of $\bmy$ and $\tilde \bmx$ to write 
\begin{align*}
 \la\bma'_i,\bmx\ra & =  \la\bma_i,\bmx\ra - 2\la\bma_i, \bmy\ra \la \bmy,\bmx\ra \\
& = \la\bma_i,\bmx\ra - 2[\bma_i]_1x_1 =- \la\bma_i, \tilde \bmx\ra  \ge 0.
\end{align*}
Because the mapping $\bma_i \to \bma'_i$  is onto and (piecewise) isometric, $\{\bma'_i\}$ will be uniformly distributed over the half sphere $\{\bma\, |\,  \la\bma,\bmx^0\ra  >0\}$ whenever $\{\bma_i\}$ are independently and uniformly distributed over the entire sphere.

Consider the ``hourglass''  shaped, symmetric set
$$\setA = \{\bma\,|\, \la\bma, \bmx\ra  > 0, \la\bma, \tilde \bmx\ra  > 0\} \cup \{\bma\, |\,  \la\bma, \bmx\ra  < 0, \la\bma, \tilde \bmx\ra < 0\} .$$
We now make the following claim:  $\setC_\Real(\bmy) \subset \bigcup_i \setC_\Real(\bma'_i)$ whenever
\aln{ \label{fullcover}
\rsphere{n} \subset \bigcup_{\bma_i\in \setA} \setC_\Real(\bma_i).
}
In words, if the caps defined by the subset of $\{\bma_i\}$ in $\setA$ cover the entire sphere, then the caps $\{\setC_\Real(\bma'_i)\}$ (which have centers in $\setC_\Real(\bmx)$) not only cover $\setC_\Real(\bmx),$ but also cover its neighbor cap $\setC_\Real(\bmy).$   To justify this claim, suppose that \eqref{fullcover} holds.   Choose some $\bmdelta \in \setC_\Real(\bmy).$  This point is covered by some cap $\setC_\Real(\bma_i)$ with $\bma_i \in \setA.$  If $\la\bma_i,\bmx\ra > 0$ and $\la\bma_i,\tilde\bmx\ra > 0,$ then $\bma_i=\bma'_i$ and $\bmdelta$ is covered by $\setC_\Real(\bma'_i).$  Otherwise, we have $\la \bma_i,\bmx\ra <0$ and $\la\bma_i,\tilde \bmx\ra < 0,$ then
\begin{align*}
\la \bmdelta, \bma'_i \ra  & = \la\bmdelta, \bma_i - 2\la\bma_i,\bmy\ra\bmy\ra \\
& = \la\bmdelta,\bma_i\ra - 2\la \bma_i,\bmy\ra\la\bmdelta, \bmy\ra \ge  \la\bmdelta,\bma_i\ra \ge 0.
\end{align*}
Note we have used the fact that $\la\bmdelta, \bmy\ra$ is real and non-negative because $\bmdelta \in \setC_\Real(y).$ We have also used  $\la\bma_i, \bmy\ra  = [\bma_i]_1 = \half(\la\bma_i, \bmx\ra+\la\bma_i, \tilde \bmx\ra)  <0,$ which follows from the definition of $ \tilde \bmx$ and the definition of~$\setA.$   In either case, $\la\bmdelta,\bma'_i\ra > 0 ,$ and we have $\bmdelta \in \setC_\Real(\bma'_i),$ which proves our claim.  

We can now see that the probability that $\setC_\Real(\bmy) \subset \bigcup_i \setC_\Real(\bma'_i)$ is at least as high as the probability that~\eqref{fullcover} holds.  Let $p_\text{cover}(m,n;\bmx,\bmy \mid m_\setA)$ denote the probability of covering $\setC(\bmy)$ conditioned on the number $m_\setA$ of points lying in $\setA.$  From \lemref{lem:coins}, we know that  
$$p_\text{cover}(m,n;\bmx,\bmy \mid m_\setA)\ge p_\text{cover}(m_\setA,n) > p_\text{cover}(m_\setA+1,n+1).$$
As noted in Lemma \ref{lem:coins}, the expression on the right is the chance of turning up more than $n$ heads when flipping $m_\setA$ fair coins.   

The probability $p_{\text{cover}}(m,n;\bmx,\bmy)$ is then given by 
\begin{align*}
p_{\text{cover}}(m,n;\bmx,\bmy) & = \mathbb{E}_{m_\setA}[p_\text{cover}(m,n;\bmx,\bmy \mid m_\setA)] \\
& \ge \mathbb{E}_{m_\setA}[p_\text{cover}(m_\setA,n)] \\
& >  \mathbb{E}_{m_\setA}[p_\text{cover}(m_\setA+1,n+1)]
.
\end{align*}
The expression on the right hand side is the probability of getting more than $n$ heads when one fair coin is flipped for every measurement $\bma_i$ that lies in $\setA$. 

Let's evaluate how often this coin-flipping event occurs.  The region $\setA$ is defined by two planes that intersect at an angle of $\beta=\ang(\bmx,\tilde \bmx) = 2\ang(\bmx,\bmy).$ The probability of a random point $\bma_i$ lying in $\setA$ is given by $\alpha=\frac{2\pi - 2\beta}{2\pi} = 1-\frac{\beta}{\pi},$ which is the fraction of the unit sphere that lies either above or below both planes.    The probability of a measurement $\bma_i$ contributing to the heads count is half the probability of it lying in $\setA,$ or $\half \alpha.$  The probability of turning up more than $n$ heads is therefore given by
$$1- \sum_{k=0}^{n} \left (\half \alpha \right)^k  \left(1-\half \alpha \right)^{m-k} {{m}\choose{k}} .$$
Using Hoeffding's inequality, we obtain the following lower bound
\begin{align*}
p_{\text{cover}}(m,n;\bmx,\bmy) & \ge 1- \sum_{k=0}^{n} \left (\half \alpha \right)^k  \left(1-\half \alpha \right)^{m-k} {{m}\choose{k}} \\
&  \ge 1 - \exp \!\left( \frac{ - (\alpha m-2n )^2 }{2m} \right )\!,
\end{align*}
which is only valid for $\alpha m>2n.$
\end{proof}

\begin{remark}
In the proof of Lemma \ref{lem:neighbor}, we obtained a bound on $ \mathbb{E}_{m_\setA}[p_\text{cover}(m_\setA+1,n+1)]$ using an intuitive argument about coin flipping probabilities. This expectation could have been obtained more rigorously (but with considerably more pain) using the method of probability generating functions.  
\end{remark}

Lemma \ref{lem:neighbor} contains most of the machinery needed for the proofs that follow.  In the sequel, we prove a number of exact reconstruction theorems for (PM).  Most of the results rely on short arguments followed by the invocation of Lemma \ref{lem:neighbor}. 

We finally state a result that bounds from below the probability of covering the sphere with caps of small central angle. The following Lemma is a direct corollary of the results of Burgisser, Cucker, and Lotz in~\cite{burgisser2010coverage}. A derivation that uses their results is given in Appendix \ref{sec:smallcap}.

\begin{lemma} \label{smallcapscover} Let $n\ge 9,$ and $m>2n.$ Then the 
probability of covering the sphere $S^{n-1}_\reals$ with independent uniformly sampled caps of central angle $\phi \le \pi/2$ is lower bounded by
\begin{align*} 
& p_\text{cover}(m,n,\phi) \ge \\
 &\qquad  1-   \frac{(em)^n\sqrt{n-1}}{(2n)^{n-1}}  \exp\!\left( \!-\frac{\sin^{n-1}(\phi)(m-n)}{\sqrt{8n}}\right ) \!\cos(\phi)
\\
 & \qquad - \exp\!\left(\! - \frac{(m-2n+1)^2}{2m-2}  \right)\!.
\end{align*}
\end{lemma}


\section{Recovery Guarantees}
 \label{sec:exact}
 
Using the uniqueness condition provided by \thmref{thm:opt} and the tools derived in \secref{sec:preliminaries}, we now develop sharp lower bounds on the success probability of PhaseMax for noiseless real- and complex-valued systems.  The noisy case will be discussed in \secref{sec:noise}.
 
\subsection{The Real Case} 

We now study problem (PM) in the case that the unknown signal and measurement vectors are real valued.
Consider some collection of measurement vectors $\{\bma_i\}$ drawn independently and uniformly from~$\rsphere{n}.$  For simplicity, we also consider the collection $\{\tilde \bma_i\}=\{\phase(\la \bma_i,\bmx^0\ra) \bma_i\}$ of aligned vectors that satisfy $\la\tilde \bma_i, \bmx^0\ra \ge 0$ for all $i$.  Using this notation, \thmref{thm:opt} can be rephrased as a simple geometric condition.

\begin{corollary} \label{cor:simplecond}
Consider the set $\{\tilde \bma_i\}=\{\phase(\la\bma_i,\bmx^0\ra)\bma_i\}$ of aligned measurement vectors.   Define the half sphere of aligned ascent directions  
$$\setD_\Real = \setC_\Real(\hat \bmx) = \{\bmdelta \in \setS^{n-1}_\Real|\,  \la \bmdelta,\hat \bmx\ra \in \reals \ge 0   \}.$$
 The true vector $\bmx^0$ will be the unique maximizer of (PM) if
$$\setD_\Real \subset \bigcup_i \setC_\Real(\tilde \bma_i).$$
\end{corollary}
\begin{proof}
Choose some ascent direction $\bmdelta \in \setD_\Real.$  If the assumptions of this Corollary hold, then there is some~$i$ with $\bmdelta \in \setC_\Real(\tilde \bma_i),$ and so $\la\tilde \bma_i,\bmdelta\ra \ge 0.$ 
Since this is true for any $\bmdelta \in \setD_\Real,$ the conditions of \thmref{thm:opt} are satisfied and exact reconstruction holds.
\end{proof}

Using this observation, we can develop the following lower bound on the success probability of PhaseMax for real-valued systems. 

\begin{theorem} \label{thm:realrecon}
Consider the case of recovering a real-valued signal $\bmx^0\in\Real^n$ from $m$ noiseless measurements of the form \eqref{eq:original} with  measurement vectors $\bma_i$, $i=1,2,\ldots,m$, sampled independently and uniformly from the unit sphere $\rsphere{n}$. Then, the probability that PhaseMax recovers the true signal~$\bmx^0$, denoted by $p_\Real(m,n)$, is bounded from below as follows:  
$$p_\reals(m,n) \ge 1 - \exp \!\left ( \frac{ - (\alpha m -2n )^2 }{2m} \right )\!,$$
  where  $\alpha= 1-\frac{2}{\pi}\ang(\bmx^0,\hat \bmx)$ and $ m>2n/\alpha.$
\end{theorem}
\begin{proof}

Consider the set of $m$ independent and uniformly sampled measurements $\{\bma_i \in \rsphere{n}\}_{i=1}^m.$  The aligned vectors $\{\tilde \bma_i = \phase(\la\bma_i,\bmx^0\ra)\bma_i\}$ are uniformly distributed over the half sphere $\setC_\Real(\bmx^0).$ Exact reconstruction happens when the condition in \corref{cor:simplecond} holds. To obtain a lower bound on the probability of this occurrence, we can simply invoke  \lemref{lem:neighbor} with $\bmx = \bmx^0$ and $\bmy = \hat \bmx.$  
\end{proof}

\subsection{The Complex Case} 

We now prove \thmref{thm:reconC} given in \secref{sec:introduction}, which characterizes the success probability of PhaseMax for phase retrieval in complex-valued systems. For clarity, we restate our result in shorter form. 

\vspace{0.3cm}
\noindent\textbf{Theorem \ref{thm:reconC}}.
{\em Consider the case of recovering a complex-valued signal $\bmx^0\in\Complex^n$ from $m$ noiseless measurements of the form \eqref{eq:original}, with $\{\bma_i\}_{i=1}^m$ sampled independently and uniformly from the unit sphere~$\csphere{n}$. 
Then, the probability that PhaseMax recovers the true signal~$\bmx^0$ is bounded from below as follows:  
\begin{align*} 
p_\Complex(m,n) \ge 1-\exp\!\left(\! - \frac{\left(\alpha m-4n\right)^2}{2m} \right)\!,
\end{align*}
where $\alpha= 1-\frac{2}{\pi}\ang( \bmx^0,\hat\bmx)$ and $ m >4n/\alpha.$}
\vspace{0.1cm}

\begin{proof}
Consider the set $\{\tilde \bma_i\}=\{\phase(\la\bma_i,\bmx^0\ra)\bma_i\}$ of aligned measurement vectors.   Define the half sphere of aligned ascent directions 
$$\setD_\Complex  = \{\bmdelta \in \csphere{n}|\,  \la \bmdelta, \hat \bmx\ra_\Re  \in \reals_0^+ \}.$$
By Theorem \ref{thm:opt}, the true signal $\bmx^0$ will be the unique maximizer of (PM) if  
\aln{ \label{simple2}
\setD_\Complex \subset \bigcup_i \setC_\Complex(\tilde \bma_i).
}
Let us bound the probability of this event.  Consider the set $\setB = \{\bmdelta\,|\, \la \bmdelta,\bmx^0\ra_\Im  = 0\}.$  We now claim that~\eqref{simple2} holds whenever
\aln{ \label{specialcover}
\setC_\Complex(\hat \bmx)\cap \setB  \subset   \bigcup_i \setC_\Complex(\tilde \bma_i).
}
To prove this claim, consider some $\bmdelta \in \setD_\Complex.$  To keep notation light, we will assume without loss of generality that  $\|\bmx^0\|_2=1.$  Form the vector $\bmdelta' = \bmdelta + j\la \bmdelta, \bmx^0 \ra_\Im \, \bmx^0,$ which is the projection of $\bmdelta$ onto $\setB$.  We can verify that  $\bmdelta' \in \setB$ by writing 
\begin{align*}
\la \bmdelta', \bmx^0 \ra  & = \la \bmdelta , \bmx^0 \ra +\la  j\la \bmdelta, \bmx^0 \ra_\Im \, \bmx^0, \bmx^0 \ra \\
& = \la \bmdelta , \bmx^0 \ra - j\la \bmdelta, \bmx^0 \ra_\Im \la  \bmx^0,\bmx^0 \ra \\
& =  \la \bmdelta ,\bmx^0 \ra - j\la \bmdelta,\bmx^0 \ra_\Im =  \la \bmdelta ,\bmx^0 \ra_\Re,
\end{align*}
which is real valued.  Furthermore, $\bmdelta' \in \setC_\Complex(\hat \bmx)$ because
$$ \la  \bmdelta', \hat \bmx \ra   =  \la \bmdelta , \hat \bmx \ra +\la  j\la \bmdelta,\bmx^0 \ra_\Im\,\bmx^0, 
\hat \bmx \ra= \la \bmdelta , \hat \bmx \ra - j\la \bmdelta,\bmx^0 \ra_\Im \la  \bmx^0, \hat \bmx \ra. $$
The first term on the right is real-valued and non-negative (because $\bmdelta \in \setD_\Complex$), and the second term is complex valued (because $\bmx^0$ is assumed to be aligned with $\hat \bmx$). It follows that $ \la  \bmdelta', \hat \bmx \ra_\Re \ge 0$ and $\bmdelta' \in \setC_\Complex(\hat \bmx).$  Since we already showed that $\bmdelta' \in \setB,$ we have $\bmdelta' \in \setC_\Complex(\hat \bmx)\cap \setB.$
Suppose now that \eqref{specialcover} holds.  The claim will be proved if we can show that  $\bmdelta\in \setD$ is covered by one of the $\setC_\Complex(\tilde \bma_i).$  Since $\bmdelta' \in \setC_\Complex(\hat \bmx)\cap \setB,$ there is some $i$ with $\bmdelta' \in \setC_\Complex(\tilde \bma_i).$  But then
 \aln{ \label{projequiv}
 0 \le \la  \bmdelta', \tilde \bma_i \ra_\Re   =  \la \bmdelta , \tilde \bma_i  \ra_\Re +\la  j\la \bmdelta,\bmx^0 \ra_\Im\, \bmx^0, 
 \tilde \bma_i  \ra_\Re=  \la \bmdelta , \tilde \bma_i  \ra_\Re. 
 }
We see that $\bmdelta \in \setC_\Complex(\tilde \bma_i),$ and the claim is proved.

We now know that exact reconstruction happens whenever condition \eqref{specialcover} holds.  We can put a bound on the frequency of this using Lemma \ref{lem:neighbor}. Note that the sphere $S^{n-1}_\comps$ is isomorphic to $S^{2n-1}_\reals,$ and the set $\setB$ is isomorphic to the sphere $S^{2n-2}_\reals.$ The aligned vectors $\{\tilde \bma_i\}$ are uniformly distributed over a half sphere in $\setC_\Complex(\bmx^0) \cap \setB,$ which is isomorphic to the upper half sphere in $S^{2n-2}_\reals.$  The probability of these vectors covering the cap $\setC_\Complex(\hat \bmx)\cap \setB$ is thus given by $p_\text{cover}(m, 2n-1;\bmx^0,
\hat \bmx)$ from Lemma \ref{lem:neighbor}. We instead use the bound for $p_\text{cover}(m, 2n;\bmx^0,\hat \bmx),$ which is slightly weaker.
\end{proof}

\begin{remark}
Theorems \ref{thm:reconC} and  \ref{thm:realrecon} guarantee exact recovery for a sufficiently large number of measurements~$m$ provided that \mbox{$\ang(\bmx^0,\hat \bmx) < \frac{\pi}{2}.$}  In the case  $\ang(\bmx^0,\hat \bmx) > \frac{\pi}{2},$ our theorems guarantee convergence to $-\bmx^0$ (which is also a valid solution) for sufficiently large $m$.  Our theorems only fail for large $m$ if $\arccos(\bmx^0,\hat \bmx) = \pi/2,$ which happens with probability zero when the approximation vector $\hat \bmx$ is generated at random. See \secref{sec:approx} for more details.
\end{remark}

\section{Generalizations}
\label{sec:generalizations}
Our theory thus far addressed the idealistic case in which the measurement vectors are independently and uniformly sampled from a unit sphere and for noiseless measurements. We now extend our results to more general random measurement ensembles and to noisy measurements. 

\subsection{Generalized Measurement Ensembles}
The theorems of Section \ref{sec:exact} require the measurement vectors $\{\bma_i\}$ to be drawn independently and uniformly from the surface of the unit sphere.  This condition can easily be generalized to other sampling ensembles.  In particular, our results still hold for all {\em rotationally symmetric distributions.}  A distribution~$D$ is rotationally symmetric if the distribution of $\bma/\|\bma\|_2$ is uniform over the sphere when $\bma \sim D.$     For such a distribution, one can make the change of variables $\bma \gets \bma/\|\bma\|_2,$ and then apply Theorems  \ref{thm:reconC} and \ref{thm:realrecon} to the resulting problem.  Note that this change of variables does not change the feasible set for (PM), and thus does not change the solution. Consequently, the same recovery guarantees apply to the original problem without explicitly implementing this change of variables. We thus have the following simple corollary.

\begin{corollary}
The results of Theorem \ref{thm:reconC} and Theorem \ref{thm:realrecon} still hold if the samples $\{\bma_i\}$ are drawn from a multivariate Gaussian distribution with independent and identically distributed (i.i.d.) entries.
\end{corollary}
\begin{proof}
A multivariate Gaussian distribution with i.i.d.\ entries is rotationally symmetric, and thus the change of variables $\bma \gets \bma/\|\bma\|_2$ yields an equivalent problem with measurements sampled uniformly from the unit sphere.  
\end{proof}

What happens when the distribution is not spherically symmetric?  In this case, we can still guarantee recovery, but we require a larger number of measurements. The following result is, analogous to \thmref{thm:reconC}, for the noiseless complex case.

\begin{theorem}
Suppose that $m_D$ measurement vectors $\{\bma_i\}_{i=1}^{m_D}$ are drawn from the unit sphere with (possibly non-uniform) probability density function $D:S^{n-1}_\comps \to \reals.$  Let $\ell_D\le \inf_{\bmx\in S^{n-1}_\comps}D(\bmx)$ be a lower bound on~$D$ over the unit sphere and let $\alpha=1-\frac{2}{\pi}\ang(\bmx^0,\hat \bmx)$ as above.  We use $s_n = \frac{2\pi^{n}}{\Gamma(n)}$ to denote the ``surface area'' of the complex sphere $\csphere{n},$ and set $m_U=\lfloor m_D s_n \ell_D \rfloor.$  Then, exact reconstruction is guaranteed with probability at least
$$1-\exp\!\left( - \frac{( \alpha m_U -4n)^2}{2m_U} \right)$$
whenever $\alpha  m_U >4n$ and $\ell_D>0$.   In other words, exact recovery with $m_D$ non-uniform measurements happens at least as often as with $m_U$ uniform measurements. 
\end{theorem}
\begin{proof}
We compare two measurement models, a uniform measurement model in which $m_U$ measurements are drawn uniformly from a unit sphere, and a non-uniform measurement model in which $m_D$ measurements are drawn from the  distribution $D.$ 
Note that the sphere $\csphere{n}$ has surface area $s_n = \frac{2\pi^{n}}{\Gamma(n)},$ and the uniform density function $U$ on this sphere has constant value $s_n^{-1}.$    Consider some collection of measurements $\{\bma_i^U\}_{i=1}^{m_U}$ drawn from the uniform model.  This ensemble of measurements lies in $\reals^{n\times m_U},$  and the probability density of sampling this measurement ensemble (in the space $\reals^{n\times m_U}$) is 
   \aln{ \label{d_uniform}
   m_U! s_n^{-m_U}.
   }
Now consider a random ensemble $\{\bma^D_i\}_{i=1}^{m_D}$ drawn with density $D.$  Given $\{\bma_i^U\}$, the event that $\{\bma_i^U\} \subset \{\bma_i^D\}$ has density (in the space $\reals^{n\times m_U}$)
\aln{\label{d_non}
 \frac{m_D!}{ (m_D-m_U)!} \prod_{i=1}^{m_U} D(\bma_i).
 }
The ratio of the non-uniform density \eqref{d_non} to the uniform density \eqref{d_uniform} is
  \aln{ \label{denseratio}
  {m_D \choose m_U} \prod_{i=1}^{m_U}  s_nD(\bma_i) \ge {m_D \choose m_U}  (s_n \ell_D)^{m_U} \ge \left(\frac{m_D s_n \ell_D}{m_U}\right)^{m_U},
  }
where we have used the bound ${m_D \choose m_U}> (m_D/m_U)^k$ to obtain the estimate on the right hand side.  The probability of exact reconstruction using the non-uniform model will always be at least as large as the probability under the uniform model, provided the ratio \eqref{denseratio} is one or higher. This holds whenever $ m_U  \le m_D s_n \ell_D .$  It follows that the probability of exact recovery using the non-uniform measurements is at least the probability of exact recovery from a uniform model with $m_U=\lfloor m_D s_n \ell_D \rfloor$ measurements.  This probability is what is given by Theorem \ref{thm:reconC}.
\end{proof}

\subsection{Noisy Measurements} \label{sec:noise}

We now analyze the sensitivity of PhaseLift to the measurement noise $\{\eta_i\}.$  For brevity, we focus only on the case of complex-valued signals. To analyze the impact of noise, we re-write the problem (PM) in the following equivalent form: 
\aln{    \label{noisy}
 \left\{\begin{array}{ll} 
\underset{\bmx\in\setH^n}{\maximize} &  \langle \bmx,\hat \bmx \rangle_\Re     \\
\st   & |\la\bma_i,\bmx\ra|^2 \le \hat b_i^2 + \eta_i, \,\, i=1,2,\ldots,m.
\end{array}\right.
}
Here, $\hat b_i^2 = |\la \bma_i,\bmx^0\ra|^2$ is the (unknown) true magnitude measurement and $b_i^2 = \hat b_i^2+\eta_i.$
We are interested in bounding the impact that these measurement errors have on the solution to (PM).  Note that the severity of a noise perturbation of size~$\eta_i$ depends on the (arbitrary) magnitude of the measurement vector~$\bma_i.$  For this reason, we assume the vectors $\{\bma_i\}$ have unit norm throughout this section.  

We will begin by proving results only for the case of non-negative noise.  We will then generalize our analysis to the case of arbitrary bounded noise.  The following result gives a geometric characterization of the reconstruction error. 

\begin{theorem} \label{noisecover}
Suppose the vectors $\{\bma_i \in \comps^n \}$ in \eqref{noisy} are normalized to have unit length, and the noise vector $\boldsymbol\eta$ is non-negative.  Let $r$ be the maximum relative noise, defined by 
\begin{align} \label{eq:relativenoisebound}
r = \max_{i=1,2,\ldots,m}  \left\{\frac{\eta_i}{\hat{b}_i}\right\}, 
\end{align}
and let $\setD_\Complex  = \{\bmdelta \in \csphere{n}|\,  \la \bmdelta, \hat \bmx\ra_\Re \ge 0 , \, \la \bmdelta, \hat \bmx \ra_\Im = 0 \}$ be the set of aligned descent directions.   
Choose some error bound $\varepsilon > r/2,$ and define the angle $\theta = \arccos(r/2\varepsilon).$ If the caps $\{\setC_\Complex(\tilde \bma_i, \theta)\}$ cover $\setD_\Complex$, then the solution $\bmx\opt$ of (PM), and equivalently of the problem in \eqref{noisy}, satisfies the bound 
 $$\|\bmx\opt -\bmx^0\|_2 \le \varepsilon.$$

\end{theorem}
\begin{proof}
We first reformulate the problem \eqref{noisy} as
\aln{    \label{errorform}
\quad \left\{\begin{array}{ll} 
\underset{\bDelta \in \comps^n}{\maximize} &  \langle\bmx^0+\bDelta,\hat \bmx \rangle_\Re     \\
\st   & |\la \tilde \bma_i,  (\bmx^0+\bDelta)\ra|^2 \le \hat b_i^2 + \eta_i, \\
&   i=1,2,\ldots,m,
\end{array}\right.
}
where $\bDelta = \bmx\opt - \bmx^0$ is the recovery error vector and $\{\tilde \bma_i\}=\{\phase(\la\bma_i,\bmx^0\ra)\bma_i\}$ are aligned measurement vectors.
In this form, the recovery error vector $\bDelta$ appears explicitly.  Because we assume the errors $\{\eta_i\}$ to be non-negative,  the true signal $\bmx^0$ is feasible for \eqref{noisy}.  It follows that the optimal objective of the perturbed problem \eqref{noisy} must be at least as large as the optimal value achieved by $\bmx^0,$ i.e., $ \langle \bDelta,\hat \bmx \rangle_\Re \ge 0$.  Furthermore, the solution $\bmx^0+\bDelta$ must be aligned with $\hat \bmx,$ as is the true signal $\bmx^0,$ and so  $ \langle \bDelta,\hat \bmx \rangle\in \reals.$  For the reasons just described, we know that the unit vector $\bmdelta =  \bDelta/\|\bDelta\|_2 \in \setD_\Complex.$

Our goal is to put a bound on the magnitude of the recovery error $\bDelta.$  We start by reformulating the constraints in \eqref{errorform} to get 
\begin{align*}
& |\la\tilde \bma_i, (\bmx^0+\bDelta)\ra|^2  = \\
& \quad |\la\tilde \bma_i,\bmx^0\ra|^2 + 2[\la  \tilde \bma_i,\bmx^0\ra^*\la \tilde \bma_i, \bDelta\ra]_\Re  + |\la\tilde \bma_i, \bDelta\ra|^2 \le \hat b_i^2 + \eta_i. 
\end{align*}
Subtracting $|\la\tilde \bma_i,\bmx^0\ra|^2=\hat b_i^2$ from both sides yields
$$  2[\la  \tilde \bma_i,\bmx^0\ra^*\la \tilde \bma_i, \bDelta\ra]_\Re  + |\la\tilde \bma_i, \bDelta\ra|^2 \le  \eta_i. $$
Since $\eta_i$ is non-negative, we have
\aln{ \label{error_bound}
\eta_i &\ge 2[\la  \tilde \bma_i,\bmx^0\ra^*\la \tilde \bma_i, \bDelta\ra]_\Re  + |\la\tilde \bma_i, \bDelta\ra|^2 \notag \\
& = 2\la  \tilde \bma_i,\bmx^0\ra \la \tilde \bma_i, \bDelta\ra_\Re  + |\la\tilde \bma_i, \bDelta\ra|^2 \notag \\
& \ge 2\la  \tilde \bma_i,\bmx^0\ra \la \tilde \bma_i, \bDelta\ra_\Re  + |\la\tilde \bma_i, \bDelta\ra_\Re|^2 . 
}
This inequality can hold only if $ \la \tilde \bma_i, \bDelta\ra_\Re$ is sufficiently small.  In particular we know
\aln{
 \la \tilde \bma_i, \bDelta\ra_\Re & \le -\la\tilde \bma_i,\bmx^0\ra+\sqrt{(\la\tilde \bma_i,\bmx^0\ra)^2+\eta_i} \notag  \\
 & \le \frac{\eta_i}{2\la\tilde \bma_i,\bmx^0\ra} = \frac{\eta_i}{2\hat b_i} \leq \frac{r}{2}. \label{noisebound1}
 }

Now suppose that $\|\bDelta\|_2 > \varepsilon.$  From \eqref{noisebound1} we have 
$$\la \tilde \bma_i, \bmdelta \ra_\Re \le  \la \tilde \bma_i, \bDelta/\|\bDelta\|_2 \ra_\Re <  \frac{r}{2\varepsilon},$$
 Therefore $\bmdelta \not \in \setC_\Complex(\tilde \bma_i, \theta)$ where $\theta = \arccos(r/2\varepsilon).$  This is a contradiction because the caps $\{C_\Complex(\tilde \bma_i, \theta)\}$ cover $\setD_\Complex.$  It follows that $\|\bDelta\|_2 \le \varepsilon.$
 \end{proof}

Using this result, we can bound the reconstruction error in the noisy case.  For brevity, we present results only for the complex-valued case.

\begin{theorem} \label{nonnegnoise}
Suppose the vectors $\{\bma_i\}$ in \eqref{noisy} are independenly and uniformly distributed in $\csphere{n},$ and the noise vector~$\boldsymbol\eta$ is non-negative.  Let $r$ be the maximum relative error defined in \eqref{eq:relativenoisebound}.  Choose some error bound $\varepsilon > r/2,$ and define the angle $\phi = \arccos(r/2\varepsilon) - \ang(\bmx^0, \hat \bmx).$  Then, the solution $\bmx\opt$ to (PM) satisfies
 $$\|\bmx\opt -\bmx^0\| \le \varepsilon.$$
with probability at least
\begin{align*}
&p_\text{cover}(m,2n,\phi) \ge \\
&\,\,\quad 1-   \frac{(em)^{2n}\sqrt{2n-1}}{(4n)^{2n-1}}  \exp\!\left( \!-\frac{\sin^{2n-1}(\phi)(m-n)}{\sqrt{16n}}\right ) \!\cos(\phi) \\
&\,\,\quad - \exp\!\left(\! - \frac{(m-4n+1)^2}{2m-2}  \right)\!
\end{align*}
when $n \ge 5 $ and $m>4n.$

\end{theorem}
     \begin{proof}

Define the following two sets:
\begin{align*}
\setD & =  \{\bmdelta \in \csphere{n} \mid  \la\bmdelta, \hat \bmx\ra \in \reals_0^+   \} \\
\setD^0  & = \{\bmdelta \in \csphere{n} \mid  \la\bmdelta,\bmx^0\ra \in \reals_0^+   \}.
\end{align*}
We now claim that the conditions of Theorem \ref{noisecover} hold whenever
\aln{ \label{specialcover2}
\setD^0   \subset   \bigcup_i \setC_\Complex(\tilde \bma_i,\phi)
}
where $\{\tilde \bma_i =\phase(\la\bma_i,\bmx^0\ra)\bma_i\}$ is the set of aligned measurement vectors. 
To prove this claim, choose some $\bmdelta \in  \setD$ and assume that \eqref{specialcover2} holds.  Since the half-sphere $\setD^0$ can be obtained by rotating $\setD$ by a principal angle of $\ang(\hat \bmx,\bmx^0),$ there is some point $\bmdelta^0 \in \setD^0$ with $\ang(\bmdelta,\bmdelta^0)\le \ang(\hat \bmx,\bmx^0).$  
By property \eqref{specialcover2}, there is some cap $\setC_\Complex(\tilde \bma_i, \phi)$ that contains~$\bmdelta_0.$  By the triangle inequality for spherical geometry it follows that:
\begin{align*}
\ang(\bmdelta,\tilde \bma_i) & \le \ang(\bmdelta,\bmdelta^0)+\ang(\bmdelta^0,\tilde \bma_i) \\
& \le \ang(\bmx^0, \hat \bmx)+\phi \le \theta.
\end{align*}
Therefore, $\bmdelta\in \setC_\Complex(\tilde \bma_i, \theta),$ and the claim is proved.

It only remains to put a bound on the probability that~\eqref{specialcover2} occurs.  Note that the aligned vectors $\{\tilde \bma_i\}$ are uniformly distributed in $\setD^0,$ which is isomorphic to a half-sphere in $S^{2n-2}_\reals.$  The probability of covering the half sphere $S^{2n-2}_\reals$ with uniformly distributed caps drawn from that half sphere is at least as great as the probability of covering the whole sphere $S^{2n-2}_\reals$ with caps drawn uniformly from the entire sphere.  This probability is given by Lemma~\ref{smallcapscover} as
$p_\text{cover}(m,2n-1,\phi),$ and is lower bounded by  $p_\text{cover}(m,2n,\phi).$
\end{proof}

We now consider the case of noise that takes on both positive and negative values.  In this case, we bound the error by converting the problem into an equivalent problem with non-negative noise, and then apply Theorem \ref{nonnegnoise}.
\begin{theorem}  \label{generalNoise}
Suppose the vectors $\{\bma_i\}$ in \eqref{noisy} are normalized to have unit length, and that $\eta_i>- \hat b_i^2$  for all $i.$\footnote{This assumption is required so that $b_i^2 = \hat b_i^2+\eta_i > 0$.} Define the following measures of the noise 
\begin{align*}
s^2 = \min_{i=1,2,\ldots,m} \left\{\frac{\hat b_i^2+\eta_i}{\hat b_i^2} \right\}= \min_{i=1,2,\ldots,m} \left\{\frac{b_i^2}{\hat b_i^2}\right\} 
\end{align*}
and
\begin{align*} \quad r = \frac{1}{s}\max_{i=1,2,\ldots,m} \left\{  \hat b_i^2 - s^2\hat b_i+\frac{\eta_i}{\hat b_i}\right\}.
\end{align*}
Choose some error bound $\varepsilon > r/2,$ and define the angle $\phi = \arccos(r/2\varepsilon) - \ang(\bmx^0, \hat \bmx).$  Then, we have the bound
$$\|\bmx\opt -\bmx^0\|_2 \le  \varepsilon + (1-s)\|\bmx^0\|_2. $$
with probability at least
\begin{align*}
& p_\text{cover}(m,2n,\phi) \ge \\
& \qquad 1-   \frac{(em)^{2n}\sqrt{2n-1}}{(4n)^{2n-1}}  \exp\!\left( \!-\frac{\sin^{2n}(\phi)(m-n)}{\sqrt{16n}}\right ) \!\cos(\phi)\\
& \qquad - \exp\!\left(\! - \frac{(m-4n+1)^2}{2m-2}  \right)\!
\end{align*}
when $n \ge 5 $ and $m>4n.$ 

\end{theorem}
     \begin{proof}
Consider the ``shrunk'' version of problem \eqref{noisy}
\aln{    \label{small}
\quad \left\{\begin{array}{ll} 
\underset{\bmx\in\setH^n}{\maximize} &  \langle \bmx,\hat \bmx \rangle_\Re     \\
\st   & |\la\bma_i,\bmx\ra|^2 \le s^2\hat b_i^2 + \zeta_i, \quad i=1,2,\ldots,m.
\end{array}\right.
}
for some real-valued ``shrink factor'' $s>0.$
Clearly, if $\bmx^0$ is aligned with $\hat \bmx$ and satisfies $|\la\bma_i,\bmx^0\ra| = b_i$ for all $i,$  then $s\bmx^0$ is aligned with $\hat \bmx$ and satisfies $|\la\bma_i,s\bmx^0\ra| = sb_i.$    We can now transform the noisy problem~\eqref{noisy} into an equivalent problem with non-negative noise by choosing
$$s^2 =  \min_{i=1,2,\ldots,m} \left\{\frac{\hat b_i^2+\eta_i }{\hat b_i^2}\right\} \quad  \text{ and } \quad \zeta_i = \hat b_i^2-s^2 \hat b_i^2+\eta_i \ge 0.$$
We then have $(sb_i)^2+\zeta_i = b_i^2 + \eta_i^2,$ and so problem \eqref{small} is equivalent to problem \eqref{noisy}.  However, the noise $\zeta_i$ in problem \eqref{small} is non-negative, and thus we can apply Theorem \ref{nonnegnoise}.  This theorem requires the constant $r$ for the shrunken problem, which is now
\begin{align*}
r_\text{shrunk}&  = \max_{i=1,2,\ldots,m} \left\{ \frac{\zeta_i}{s\hat b_i} \right\}  \\
&  =  \frac{1}{s}\max_{i=1,2,\ldots,m} \left\{  {\hat b_i^2 - s^2\hat b_i+\eta_i/\hat b_i} \right\}\!.
\end{align*}
The solution to the shrunk problem \eqref{small} satisfies $\|\bmx\opt - s\bmx^0\|_2 \le \epsilon,$ with probability $p_\text{cover}(m,2n,\phi),$ where $\phi = \arccos(r_\text{shrunk}/2\epsilon) - \ang(\bmx^0, \hat \bmx).$
 If this condition is fulfilled, then we have
\begin{align*}
\|\bmx\opt -\bmx^0\|_2 & \le  \|\bmx\opt - s\bmx^0 +s\bmx^0   -\bmx^0\|_2  \\
 & \le \|\bmx\opt - s\bmx^0\|_2 +\|s\bmx^0   -\bmx^0\|_2 \\
 & \le  \epsilon + (1-s)\|\bmx^0\|_2, 
 \end{align*}
which concludes the proof.
\end{proof}

\begin{remark}
Theorem \ref{generalNoise} requires the noise to be sufficiently small so that the measurements are non-negative, and the PhaseMax formulation is feasible.
A natural extension that avoids this caveat is to enforce constraints with a hinge penalty rather than a hard constraint.   In addition, it is possible to achieve better noise robustness in this case.  This direction was studied in \cite{hand2016corruption}. 
\end{remark}


\section{How to Compute Approximation Vectors?}
\label{sec:approx}
There exist a variety of algorithms that compute approximation vectors\footnote{Approximation vectors are also known as initialization or anchor vectors.}, such as the (truncated) spectral initializer \cite{netrapalli2013phase,chen2015solving} or corresponding optimized variants \cite{lu2017phase,mondelli2017fundamental},  the Null initializer~\cite{chen2015phase}, the orthogonality-promoting method~\cite{wang2016solving}, or least-squares methods \cite{bendory2016non}. 
We now show that even randomly generated approximation vectors guarantee the success of PhaseMax with high probability given a sufficiently large number of measurements. 
We then show that more sophisticated methods guarantee success with high probability if the number of measurements depends linearly on~$n.$


\subsection{Random Initialization}

Consider the use of approximation vectors $\hat\bmx$ drawn randomly from the unit sphere $S^{n-1}_\reals.$
Do we expect such approximation vectors to be accurate enough to recover the unknown signal?  To find out, we analyze the inner product between two real-valued random vectors on the unit sphere.  Note that we only care about the {\em magnitude} of this inner product. If the inner product is negative, then PhaseMax simply recovers $-\bmx^0$ rather than $\bmx^0$. Our analysis will make use of the following result. 

\begin{lemma} \label{lem:innerproductdistr}
Consider the angle $\beta = \ang(\bmx,\bmy)$ between two random vectors $\bmx,\bmy\in\hsphere{n}$ sampled independently and uniformly from the unit sphere. Then, the expected {\em magnitude} of the cosine distance between the two random vectors satisfies
%
\begin{align}
\sqrt{\frac{2}{\pi n}} \le \mathbb{E}[|\cos(\beta)|]  \le \sqrt{\frac{2}{\pi(n-\frac{1}{2})}},&\,\, \text{ for } \,\setH=\reals \label{eq:realboundX}\\ 
 \sqrt{\frac{1}{\pi n}} \le \mathbb{E}[|\cos(\beta)|]  \le \sqrt{\frac{4}{\pi(4n-1)}}, &\,\, \text{ for }\, \setH=\comps.  \label{eq:complexboundX}
\end{align}
\end{lemma}

\begin{proof}
We first consider the real case.  The quantity $\cos(\beta) =\la \bmx,\bmy \ra / (\|\bmx\|_2\|\bmy\|_2)$ is simply the sample correlation between two random vectors, whose distribution function is given by~\cite{kenney1951}
\aln{ \label{cosdist}
f(z) = \frac{(1-z^2)^{\frac{n-3}{2}}}{2^{n-2}B\!\left(\frac{n-1}{2},\frac{n-1}{2}\right)},
}
where $B$ is the beta function and $z\in[-1,+1]$.  
Hence, the expectation of the magnitude of the inner product is  given by
$$\mathbb{E}[|\cos(\beta)|] = 2  \frac{\int_0^1z(1-z^2)^{\frac{n-3}{2}} \, \text{d}z}{2^{n-2}B\!\left(\frac{n-1}{2},\frac{n-1}{2}\right)} .
$$
The integral in the numerator was studied in \cite[Eq.~31]{jacques2015quantized} and evaluates to  $\half B(1, \frac{n-1}{2}).$  Plugging this expression into \eqref{cosdist}, and using the identity $\Gamma(a)/\Gamma(a/2) = 2^{a-1}\Gamma(\frac{a+1}{2})/\sqrt{\pi},$  we get
   $$\mathbb{E}[|\cos(\beta)|] = \frac{\Gamma(\frac{n}{2})}{\sqrt{\pi}\,\Gamma(\frac{n+1}{2})}.$$
Finally, by using bounds on ratios of Gamma functions~\cite{qi2012bounds}, we obtain the bounds in \eqref{eq:realboundX} for real-valued vectors.  The bounds in \eqref{eq:complexboundX} for complex-valued vectors are obtained by noting that $\csphere{n}$ is isomorphic to~$\rsphere{2n}$ and by simply replacing $n\gets2n$ in the bounds for the real-valued case.
   \end{proof}

For such randomly-generated approximation vectors, we now consider the approximation accuracy $\alpha$ that appears in Theorems~\ref{thm:reconC} and~\ref{thm:realrecon}. Note that $\mathbb{E}[|\beta|] \le \frac{\pi}{2} - \mathbb{E}[|\cos(\beta)|],$ and, thus
$$\mathbb{E}[\alpha] = 1 - \frac{2}{\pi} \mathbb{E}[|\beta|]  \ge  \frac{2}{\pi} \mathbb{E}[|\cos(\beta)|] \ge \sqrt{\frac{8}{\pi^3n}}$$
for the real case. 
Plugging this bound on the expected value for $\alpha$ into Theorem \ref{thm:realrecon}, we see that, for an {\em average} randomly-generated approximation vector (one with $\alpha > \sqrt{8/(\pi^3n)}$), the probability of exact reconstruction goes to $1$ rapidly as $n$ goes to infinity, provided that  the number of measurements satisfies
$ m >  c n^{3/2}$ for any  $c > \sqrt{\pi^3/2}.$   For complex-valued signals and measurement matrices, this becomes $c > 2\sqrt{\pi^3}.$ 
  

%
%
Our results indicate that the use of random approximation vectors requires $O(n^{3/2})$ measurements rather than the $O(n)$ required by other phase retrieval algorithms with other initialization methods, e.g., the ones proposed in~\cite{candes2014solving,wang2016solving}  (see also \secref{sec:discussion}). 
Hence, it may be more practical for PhaseMax to use approximation vectors obtained from more sophisticated initialization algorithms.


\subsection{Spectral Initializers}
\revised{Spectral initializers were first used in phase retrieval by~\cite{netrapalli2013phase}, and enable the computation of an approximation vector $\hat\bmx$ that exhibits strong theoretical properties~\cite{chen2015solving}.  This class of initializers was analyzed in detail in \cite{lu2017phase}, and \cite{mondelli2017fundamental} developed an optimal spectral initializer for a class of problems including phase retrieval. 
}

\revised{Fix some $\delta>1,$ and consider a measurement process with $m=\delta n$ measurement vectors sampled from a Gaussian distribution. Then, the spectral initializers in~\cite{lu2017phase,mondelli2017fundamental} are known to deliver an initializer $\hat \bmx$ with $\lim_{n\to \infty} \ang(\hat \bmx, \bmx^0)>\epsilon$ for some positive $\epsilon.$  Equivalently, there is some accuracy parameter $\alpha_\star = \lim_{n\to\infty} 1-\frac{2}{\pi}\ang(\hat \bmx, \bmx^0) > 0$.
By combining this result with \thmref{thm:reconC}, we see that spectral initializers enable PhaseMax to succeed for large $n$ with high probability provided that $m> n/\alpha_*$ measurements are used for signal recovery\footnote{Our results assume the initial vector is independent of the measurement vectors.  The total number of measurements needed is thus $m> n/\alpha_* + \delta n$ for initialization and recovery combined.};  a number of measurements that is linear in $n$.  Note that similar non-asymptotic guarantees can be achieved for finite $n$ using other results on spectral initializers (e.g., Prop. 8 in \cite{chen2015solving}), although with less sharp bounds.
}

\revised{In Section \ref{sec:tight}, we investigate the numerical value of~$\alpha_*$ and make quantitative statements about how many measurements are required when using a spectral initializer.  As we will see there, PhaseMax requires roughly $m=5.5n$ noiseless Gaussian measurements in theory\footnote{The constants that appear in the bounds \cite{lu2017phase,mondelli2017fundamental} were approximated using stochastic numerical methods, and were not obtained exactly using analysis.}, and somewhat fewer measurements in practice.
}


\section{Discussion}
\label{sec:discussion}
This section briefly compares our theoretical results to that of existing algorithms. We furthermore demonstrate the sharpness of our recovery guarantees  and show the practical limits of PhaseMax. 

\subsection{Comparison with Existing Recovery Guarantees}
\revised{ \tblref{tbl:comparison} compares our noiseless recovery guarantees in a complex system to that of PhaseLift \cite{candes2014solving}, truncated Wirtinger flow (TWF)~\cite{candes2014solving}, and truncated amplitude flow (TAW) \cite{wang2016solving}.  We also compare to the recovery guarantee provided for PhaseMax using classical machine learning methods in \cite{bahmani2017phase}.  \footnote{Since AltMinPhase \cite{netrapalli2013phase} requires an online measurement model that differs significantly from the other algorithms considered here, we omit a comparison.}
We see that PhaseMax requires the same sample complexity (number of required measurements) as compared to PhaseLift, TWF, and TAW, when used together with the truncated spectral initializer \cite{candes2014solving}. While the constants $c_0$, $c_1$, and~$c_2$ in the recovery guarantees  for all of the other methods are generally very large, our recovery guarantees contain no unspecified constants, explicitly depend on the approximation factor $\alpha$, and are surprisingly sharp. We next demonstrate the accuracy of our results via numerical simulations.
}


\begin{table*}[tp]
\centering
\caption{Comparison of Theoretical Recovery Guarantees for Noiseless Phase Retrieval}
\label{tbl:comparison}
\begin{tabular}{@{}llll@{}}
\toprule
Algorithm & Sample complexity & Lower bound on $p_\mathbb{C}(m,n)$  \\
\midrule
PhaseMax &$m> 4n/\alpha$ & $1-e^{- (\alpha m-4n)^2/(2m)}$  \\[0.19cm]
PhaseLift \cite{candes2014solving} & $m\geq c_0 n$ & $1-c_1e^{-c_2 m}$     \\[0.19cm]
TWF \cite{candes2014solving} & $m\geq c_0 n$ & $1-c_1e^{-c_2 m}$    \\[0.19cm]
TAF \cite{wang2016solving} & $m\geq c_0 n$ & $1-(m+5)e^{-n/2}-c_1e^{-c_2 m}-1/n^2$  \\  
Bahmani and Romberg  \cite{bahmani2017phase} & $m> \frac{32}{\sin^4(\alpha)}\log\!\big(\frac{8e}{\sin^4(\alpha)}\big)n$ & $1- 8 e^{-\sin^4(\alpha)\big(M  -  \frac{32}{\sin^4(\alpha)}\log\!\big(\frac{8e}{\sin^4(\alpha)}\big)N\big)/16}$  \\
\bottomrule
\end{tabular}
\end{table*}

\subsection{Tightness of Recovery Guarantees} \label{sec:tight}
We now investigate the tightness of the recovery guarantees of PhaseMax using both experiments and theory.  First, we compare the empirical success probability of PhaseMax in a noiseless and complex-valued scenario with measurement vectors taken independently and uniformly from the unit sphere. 
All experiments were obtained by using the implementations provided in the software library PhasePack \cite{goldstein2017phasepack}.
This software library provides efficient implementations of phase retrieval methods using fast adaptive gradient solvers~\cite{GoldsteinStuderBaraniuk:2014}.  We declare signal recovery a success whenever the relative reconstruction error satisfies
\aln{
\textit{RRE}=\frac{\|\bmx^0-\bmx\|^2_2}{\|\bmx^0\|^2_2} < 10^{-5}.
}
We compare empirical rates of success to the theoretical lower bound in \thmref{thm:reconC}. 
\figref{fig:comparison} shows results for $n=100$ and $n=500$ measurements, where we artificially generate an approximation $\hat\bmx$ for different angles $\beta=\ang(\hat \bmx,\bmx^0)$ measured in degrees. Clearly, our theoretical lower bound accurately predicts the real-world performance of PhaseMax.  For large $n$ and large $\beta$, the gap between theory and practice becomes extremely tight.   We furthermore observe a sharp phase transition between failure and success, with the transition getting progressively sharper for larger dimensions $n$.

\begin{figure*}[tp]
\centering
\subfigure[$n=100$]{\includegraphics[width=0.48\textwidth]{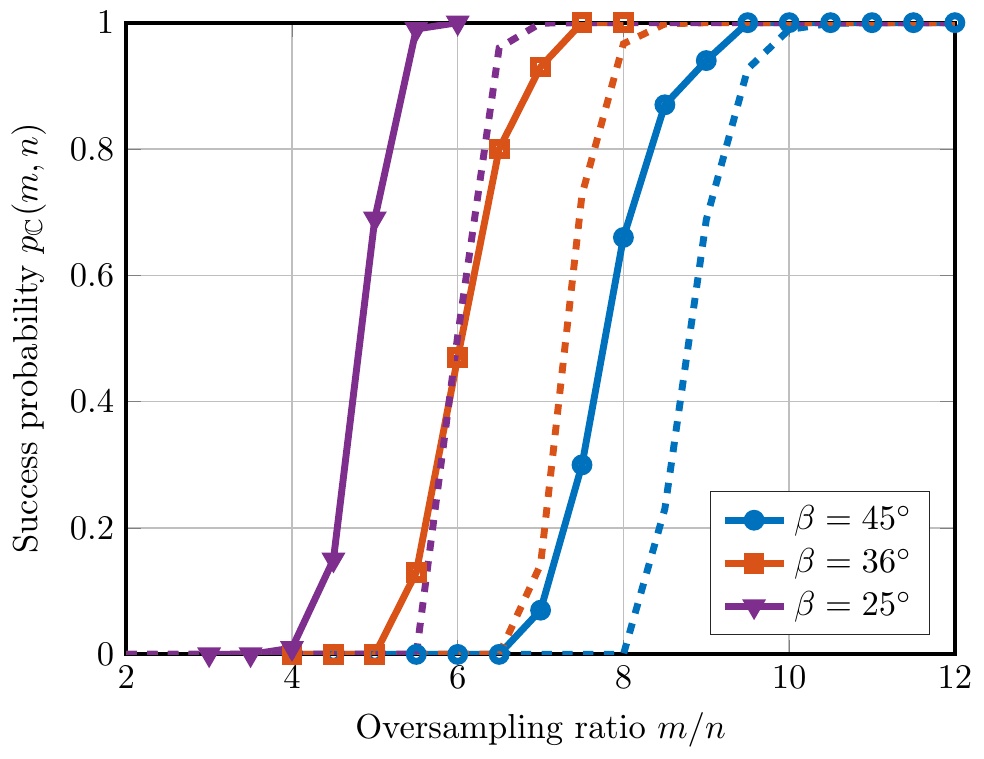}}
\hspace{0.2cm}
\subfigure[$n=500$]{\includegraphics[width=0.48\textwidth]{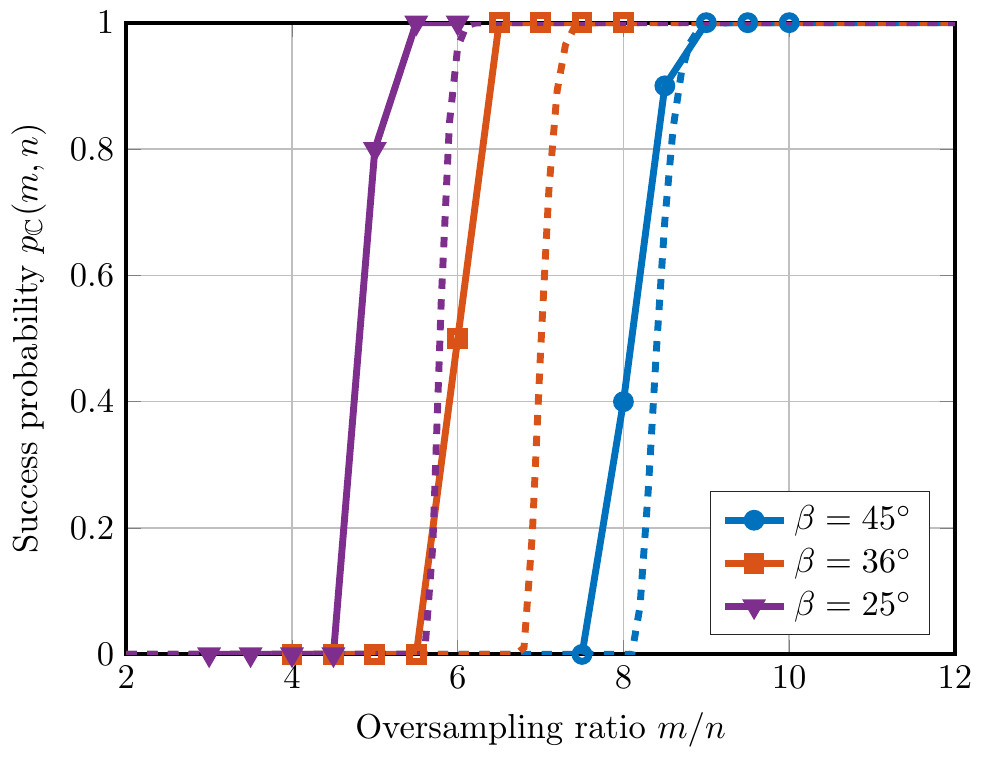}}
\caption{Comparison between the empirical success probability (solid lines) and our theoretical lower bound (dashed lines) for varying angles $\beta$ between the true signal and the approximation vector. Our theoretical results accurately characterize the empirical success probability of PhaseMax. Furthermore, PhaseMax exhibits a sharp phase transition for larger dimensions.}
\label{fig:comparison}
\end{figure*}

\revised{
Next, we investigate the tightness of the recovery guarantees when an initial vector is chosen using the spectral initialization method \cite{lu2017phase,mondelli2017fundamental}.  Using analytical formulas, the authors of~\cite{mondelli2017fundamental} calculate the accuracy of their proposed spectral initializer for real-valued signal estimation from Gaussian measurements with large $n$.  They find that, when $2n$ measurements are used for spectral estimation, the initializer and signal have a squared cosine similarity of over $0.6.$ This corresponds to an accuracy parameter of $\alpha > 0.55,$ and PhaseMax can recover the signal exactly with $m=n/\alpha\approx 3.5 n$ measurements.  This is within a factor of 2 of the information-theoretic lower bound for real-valued signal recovery, which is $2n-2$ measurements.
}

\revised{Note that our analysis of PhaseMax assumes that the measurements used for recovery are statistically independent from those used by the initializer.  For this reason, our theory does not allow us to perform signal recovery by ``recycling'' the measurements used for initialization.  While we do not empirically observe any change in behavior of the method when the initialization measurements are used for recovery, the above reconstruction bounds formally require $5.5n$ measurements ($2n$ for initialization, plus $3.5n$ for recovery).  The bounds in \cite{bahmani2017phase} are uniform with respect to the initializer, and thus enable measurement recycling (although this does not result in tighter bounds for the overall number of measurements because the required constants are larger).}

\revised{
Finally, we would like to mention that asymptotically {\em exact} performance bounds for PhaseMax have recently been derived in \cite{dhifallah2017phase}, and tighter bounds for non-lifting phase retrieval have been shown for the method PhaseLamp, which repeatedly uses PhaseMax within an iterative process.
}

\begin{figure*}[tp]
\centering
\subfigure[$n=100$]{\includegraphics[width=0.48\textwidth]{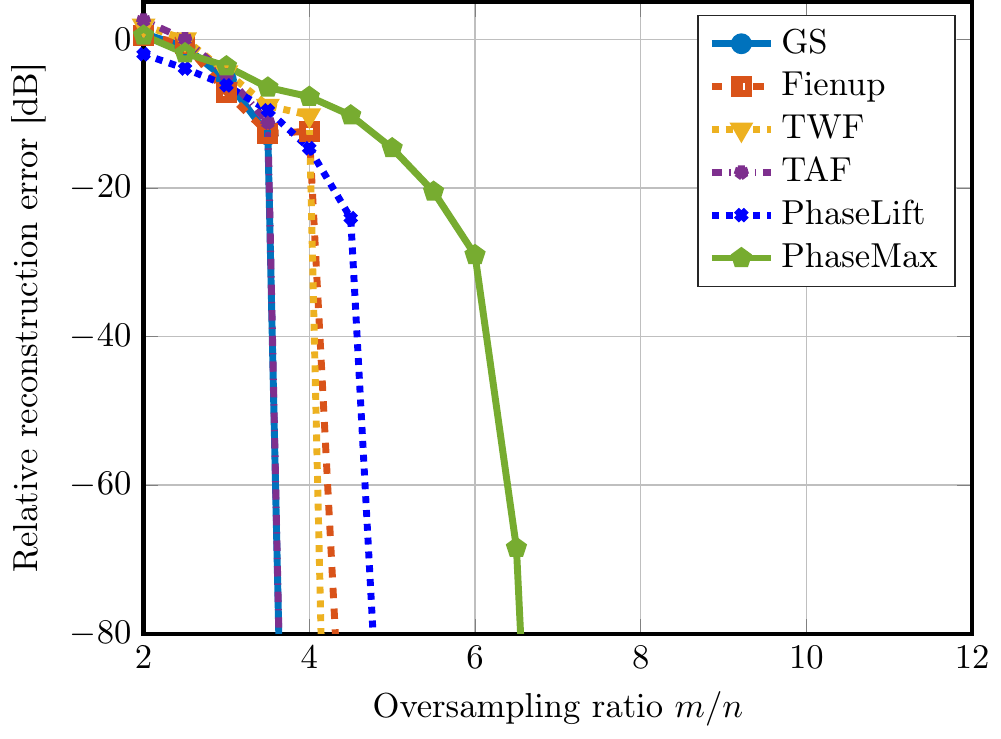}}
\hspace{0.2cm}
\subfigure[$n=500$]{\includegraphics[width=0.48\textwidth]{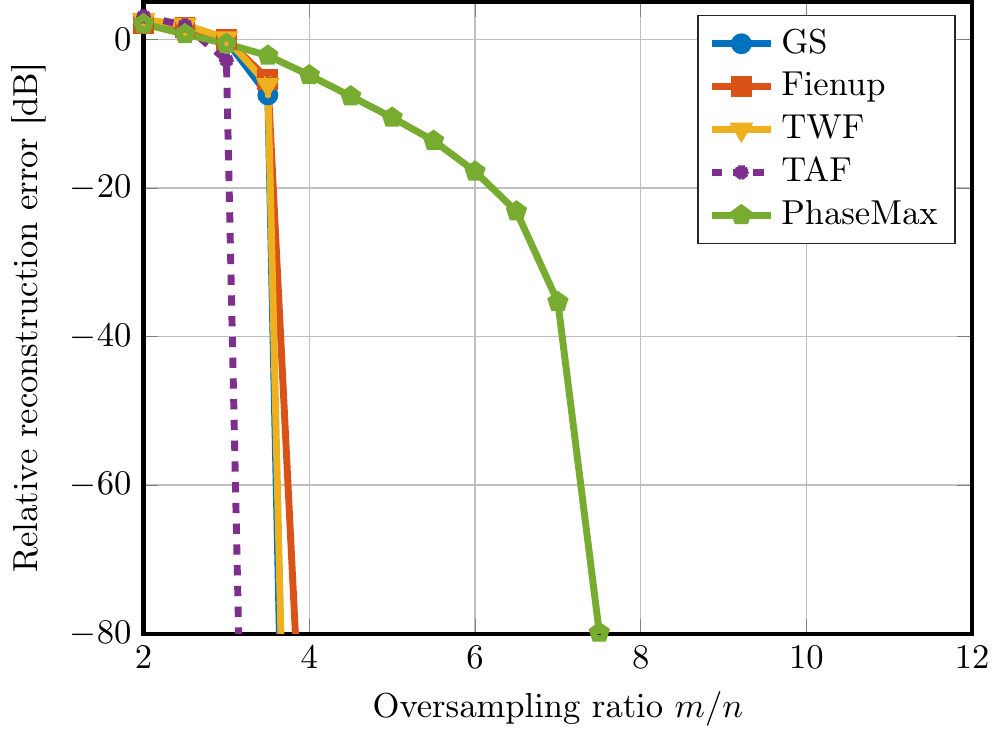}}
\caption{Comparison of the relative reconstruction error. We use the truncated spectral initializer for Gerchberg-Saxton (GS), Fienup, truncated Wirtinger flow (TWF), truncated amplitude flow (TAF), and PhaseMax.  PhaseMax does not achieve exact recovery for the lowest number of measurements among the considered methods, but is convex, operates in the original dimension, and comes with sharp performance guarantees. PhaseLift only terminates in reasonable computation time for $n=100$. }
\label{fig:algorithms}
\end{figure*}

\begin{figure*}[h]
\centering
\begin{minipage}{2.9cm} \centering
Original\vspace{13.5pt}
\includegraphics[width=2.9cm]{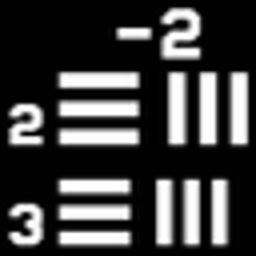}
\end{minipage}
\begin{minipage}{2.9cm} \centering
Gerchberg-Saxton \\ 0.13019\\
\vspace{4pt}
\includegraphics[width=2.9cm]{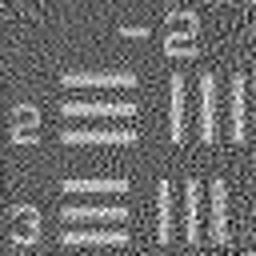}
\end{minipage}
\begin{minipage}{2.9cm} \centering
Wirtinger Flow \\ 0.13930\\ 
\vspace{4pt}
\includegraphics[width=2.9cm]{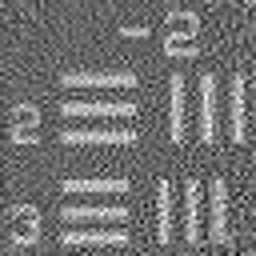}
\end{minipage}
\begin{minipage}{2.9cm} \centering
Trunc. Wirt. Flow\\0.18017\\
\vspace{4pt}
\includegraphics[width=2.9cm]{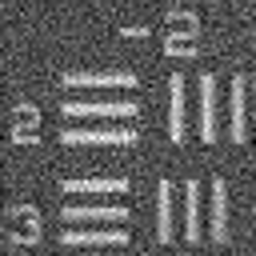}
\end{minipage}
\begin{minipage}{2.9cm} \centering
PhaseMax\\0.23459\\
\vspace{4pt}
\includegraphics[width=2.9cm]{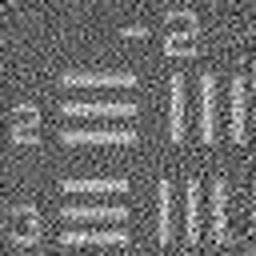}
\end{minipage}
\begin{minipage}{2.9cm} \centering
PhaseLift\\0.35452\\
\vspace{4pt}
\includegraphics[width=2.9cm]{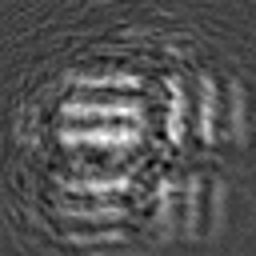}
\end{minipage}


\begin{minipage}{2.9cm} \centering
\vspace{13.5pt} 
\includegraphics[width=2.9cm]{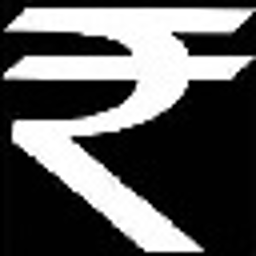}
\end{minipage}
\begin{minipage}{2.9cm} \centering
\vspace{4pt}
0.12468  \vspace{3pt}\\
\includegraphics[width=2.9cm]{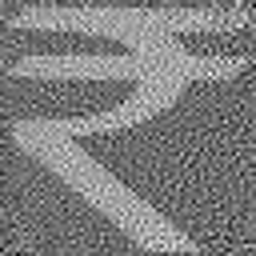}
\end{minipage}
\begin{minipage}{2.9cm} \centering
\vspace{2pt}
0.13386 \vspace{3pt}\\
\includegraphics[width=2.9cm]{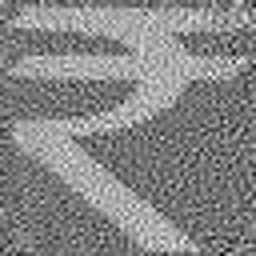}
\end{minipage}
\begin{minipage}{2.9cm} \centering
\vspace{2pt}
0.17836 \vspace{3pt}\\
\includegraphics[width=2.9cm]{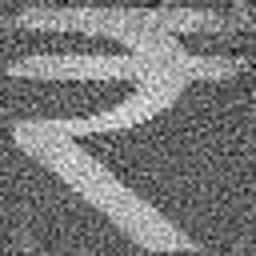}
\end{minipage}
\begin{minipage}{2.9cm} \centering
\vspace{2pt}
0.24947 \vspace{3pt}\\
\includegraphics[width=2.9cm]{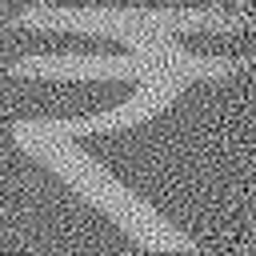}
\end{minipage}
\begin{minipage}{2.9cm} \centering
\vspace{2pt}
0.35282 \vspace{3pt}\\
\includegraphics[width=2.9cm]{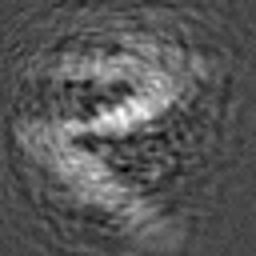}
\end{minipage}

\caption{ Reconstruction of two $64\times64$ masks from empirical phaseless measurements obtained through a diffusive medium \cite{metzler2017coherent}. The numbers on top of each of the recovered images denote the relative measurement error $\||\bA\bmx|-\bmb\|_2/\|\bmb\|_2$ achieved by each method.} 
\label{fig:aft}
\end{figure*}

\subsection{Experimental Results}
We compare PhaseMax to other algorithms using synthetic and empirical datasets.  The implementations used here are publicly available as part of the PhasePack software library~\cite{goldstein2017phasepack}.  This software library contains  scripts for comparing different phase retrieval methods, including scripts for reproducing the experiments shown here.

\subsubsection{Comparisons Using Synthetic Data}
We briefly compare PhaseMax to a select set of phase retrieval algorithms in terms of relative reconstruction error with Gaussian measurements. 
We emphasize that this comparison is by no means intended to be exhaustive and serves the sole purpose of demonstrating the efficacy and limits of PhaseMax  (see, e.g., \cite{waldspurger2015phase,jaganathan2015phase} for more extensive phase retrieval algorithm comparisons). We compare the Gerchberg-Saxton algorithm~\cite{gerchberg1972practical}, the Fienup algorithm~\cite{fienup1982phase},  truncated Wirtinger flow~\cite{chen2015solving}, and PhaseMax---all of these methods use the truncated spectral initializer~\cite{chen2015solving}. 
We also run simulations using the semidefinite relaxation (SDR)-based method PhaseLift~\cite{candes2013phaselift} implemented via FASTA~\cite{GoldsteinStuderBaraniuk:2014}; this is, together with PhaseCut~\cite{waldspurger2015phase}, the only convex alternative to PhaseMax, but lifts the problem to a higher dimension. 

\figref{fig:algorithms} reveals that PhaseMax requires larger oversampling ratios $m/n$ to enable faithful signal recovery compared to non-convex phase-retrieval algorithms that operate in the original signal dimension. This is because the truncated spectral initializer requires oversampling ratios of about six or higher to yield sufficiently accurate approximation vectors~$\hat\bmx$ that enable PhaseMax to succeed.

\subsubsection{Comparisons Using Empirical Data}
The PhasePack library contains scripts for comparing phase retrieval algorithms using synthetic measurements as well as publicly available real datasets.  The datasets contain measurements obtained using the experimental setup described in \cite{metzler2017coherent}, in which a binary mask is imaged through a diffusive medium.

Figure \ref{fig:aft} shows reconstructions from measurements obtained using two different test images.  Both images were acquired using phaseless measurements from the same measurement operator $\bA$.  Reconstructions are shown for the Gerchberg-Saxton~\cite{gerchberg1972practical}, Wirtinger flow \cite{candes2015wirtinger}, and truncated Wirtinger flow~\cite{chen2015solving} methods.  We also show results of the two convex methods PhaseLift (which lifts the problem dimension) \cite{candes2013phaselift} and PhaseMax (the proposed non-lifting relaxation).  For each algorithm, Figure \ref{fig:aft} shows the recovered images and reports the relative measurement error $\||\bA\bmx|-\bmb\|_2/\|\bmb\|_2.$
All phase retrieval algorithms were initialized using a variant of the spectral method proposed in \cite{mondelli2017fundamental}.

We find that the Gerchberg-Saxton algorithm outperforms all other considered methods by a small margin (in terms of measurement error) followed by the Wirtinger flow.  On this dataset, measurements seem to be quite valuable, and the truncated Wirtinger flow method (with default truncation parameters as described in \cite{chen2015solving}) performs slightly worse than the original Wirtinger flow, which uses the full dataset.  PhaseMax produces results that are visually comparable to that of Wirtinger flow, but with slightly higher relative measurement error.

\subsection{Advantages of PhaseMax}
While PhaseMax does not achieve exact reconstruction with the lowest number of measurements, it is convex, operates in the original signal dimension, can be implemented via efficient solvers for Basis Pursuit,  and comes with {\em sharp} performance guarantees that do not sweep constants under the rug (cf.~\figref{fig:comparison}).
The convexity of PhaseMax enables a natural extension to sparse phase retrieval \cite{jaganathan2013sparse,shechtman2014gespar} or other signal priors (e.g., total variation or bounded infinity norm) that can be formulated with convex functions.  Such non-differentiable priors cannot be efficiently minimized using simple gradient descent methods (which form the basis of Wirtinger or amplitude flow, and many other methods), but can potentially be solved using standard convex solvers when combined with the PhaseMax formulation.


\section{Conclusions}
\label{sec:conclusions}
We have proposed a novel, convex phase retrieval algorithm, which we call \emph{PhaseMax}.
We have provided accurate bounds on the success probability that depend on the signal dimension, the number of measurements, and the angle between the approximation vector and the true vector. 
Our analysis covers a broad range of random measurement ensembles and characterizes the impact of general measurement noise on the solution accuracy.
We have demonstrated the sharpness of our recovery guarantees and studied the practical limits of PhaseMax via simulations. 
%

%
There are many avenues for future work. 
Developing a computationally efficient algorithm for solving PhaseMax (or the related PhaseLamp procedure) accurately and for large problems is a challenging problem. 
An accurate analysis of more general noise models is left for future work.



\appendices

\section{Proof of Lemma \ref{realcut}} \label{sec:cutproof}

The proof is by induction.   As a base case, we note that $r(n,1)=2$ for  $n\ge 1$ and $r(2,k) = 2k$ for $k \ge1.$ 
Now suppose we have a sphere $\rsphere{n}$ in $n$ dimensions sliced by $k-1$ planes into $r(n,k-1)$  ``original'' regions.  Consider the effect of adding a $k$th plane,  $\setP_k$.  Every original region that is intersected by $\setP_k$ will split in two, and increase the number of total regions by 1.  The increase in the number of regions is then the number of original regions that are intersected by $\setP_k$.  Equivalently, this is the number of regions formed inside $\setP_k$ by the original $k-1$ planes.  Any subset of $k-1$ planes will have normal vectors that remain linearly independent when projected into $\setP_k.$ By the induction hypothesis, the number of regions formed inside $\setP_k$ is then given by $r(n-1,k-1)$.  Adding this to the number of original regions yields  $$r(n,k)=r(n,k-1)+r(n-1,k-1).$$  We leave it to the reader to verify that \eqref{numslices} satisfies this recurrence relation and base case.

\section{Proof of Lemma \ref{smallcapscover}} \label{sec:smallcap}

In this section, we prove Lemma \ref{smallcapscover}.  This Lemma is a direct corollary of the following result of Burgisser, Cucker, and Lotz \cite{burgisser2010coverage}.  For a complete proof of this result, see Theorem 1.1 of \cite{burgisser2010coverage}, and the upper bound on the constant ``C'' given in Proposition~5.5.

\begin{theorem} \label{smallcapscover_nasty}   Let $m>n\ge 2.$ Then the 
probability of covering the sphere $S^{n-1}_\reals$ with independent and uniform random caps of central angle $\phi \le \pi/2$ is bounded by
\begin{align*} 
& p_\text{cover}(m,n,\phi) \ge \\
& \qquad 1-   {{m}\choose{n}} C\int_{0}^\epsilon (1-t^2)^{(n^2-2n-1)/2} (1-\lambda(t))^{m-n} \, \text{d} t \\
&  \qquad - \frac{1}{2^{m-1}}\sum_{k=0}^{n-1}   {{m-1}\choose{k}}
\end{align*}
where $\lambda(t) = \frac{V_{n-1}}{V_n} \int_0^{\arccos(t)} \sin^{n-2}(\phi)\, \text{d}\phi,$ 
 $V_n=\textit{Vol}(S^{n-1}_\Real) = \frac{2\pi^{n/2}}{\Gamma(n/2)},$
  $C = \frac{n\sqrt{n-1}}{2^{n-1}},$
  and $\epsilon = \cos(\phi).$
\end{theorem}

While Theorem \ref{smallcapscover_nasty} provides a bound on $p_\text{cover}(m,n,\phi),$ the formulation of this bound does not provide any intuition of the scaling of   $p_\text{cover}(m,n,\phi)$ or its dependence on $m$ and $n.$  For this reason, we derive Lemma \ref{smallcapscover}, which is a weaker but more intuitive result.  We restate Lemma \ref{smallcapscover} here for clarity.

\vspace{3mm}
\noindent \textbf{Lemma \ref{smallcapscover}.}   Let $n\ge 9,$ and $m>2n.$ Then, the 
probability of covering the sphere $S^{n-1}_\reals$ with caps of central angle $\phi \le \pi/2$ is lower bounded by
\begin{align*}
& p_\text{cover}(m,n,\phi) \ge \\
& \quad 1-   \frac{(em)^n\sqrt{n-1}}{(2n)^{n-1}}  \exp\!\left(\! -\frac{sin^{(n-1)}(\phi)(m-n)}{\sqrt{8n}}\right )\! \cos(\phi) \\
& \quad - \exp\!\left(\! - \frac{(m-2n+1)^2}{2m-2}  \right)\!.
\end{align*}

\begin{proof} 
Let us simplify the result of Theorem \ref{smallcapscover_nasty}. If we assume $m > 2n,$ then Hoeffding's inequality yields
$$\frac{1}{2^{m-1}}\sum_{k=0}^{n-1}   {{m-1}\choose{k}} \le \exp\left( - \frac{(m-2n+1)^2}{2m-2}  \right)\!.$$

Next, we derive a lower bound as follows:
 \alns{
 \lambda(t)   &=  \frac{\Gamma(n/2)}{\Gamma((n-1)/2)\sqrt{\pi}} \int_0^{\arccos(t)} \sin^{n-2}(\phi)\, \text{d}\phi \\
   &\ge  \sqrt{(n/2-1)/\pi}  \int_0^{\arccos(t)} \sin^{n-2}(\phi)\cos(\phi)\, \text{d}\phi \\
       &=   \sqrt{(n/2-1)/\pi} \frac{1}{n-1} \sin^{n-1}\arccos(t)\\
       & \ge \frac{1}{\sqrt{8n}}(1-t^2)^{(n-1)/2}. 
 }
We have used the fact that $ \sqrt{(n/2-1)/\pi} \frac{1}{n-1}  > \frac{1}{\sqrt{8n}}$ for $n\ge4,$ and also the ``Wallis ratio'' bound $\frac{\Gamma(n/2)}{\Gamma((n-1)/2)} \ge \sqrt{n/2-1}$ \cite{mortici2010new,gautschi1959some}.
Finally, we plug in the inequality  ${{m}\choose{n}} \le \frac{(em)^n}{n^n}$. We now have
\begin{align*}
&{{m}\choose{n}} C\int_{0}^\epsilon (1-t^2)^{(n^2-2n-1)/2} (1-\lambda(t))^{m-n} \, \text{d} t \\ 
&\le  \frac{(em)^n\sqrt{n-1}}{(2n)^{n-1}}\int_{0}^\epsilon (1-t^2)^{(n^2-2n-1)/2} \\
& \quad \qquad \times  \left(1- \frac{1}{\sqrt{8n}}(1-t^2)^{(n-1)/2}\right)^{m-n} \, \text{d} t.
\end{align*}

Now we simplify the integral.  Using the identity $(1-x)^a < e^{-ax},$ which holds for $x\le1,$ we can convert each term in the integrand into an exponential. We do this first with $x=t^2$ and then with $x = \frac{1}{\sqrt{8n}}(1-t^2)^{(n-1)/2}$ to obtain
\begin{multline}
 (1-t^2)^{(n^2-2n-1)/2} \left(1- \frac{1}{\sqrt{8n}}(1-t^2)^{(n-1)/2}\right)^{m-n} \\ 
 \le \exp\!\left(\!-\frac{t^2(n^2-2n-1)}{2}  -\frac{(1-t^2)^{(n-1)/2}(m-n)}{\sqrt{8n}}\right )\!.
 \end{multline}
 We then apply the Cauchy-Schwarz inequality to get
 \begin{align*}
 & \int_{0}^\epsilon \exp\left(-\frac{t^2(n^2-2n-1)}{2}  -\frac{(1-t^2)^{(n-1)/2}(m-n)}{\sqrt{8n}}\right )  \text{d}t \\
 &\qquad  \le  \left[\int_{0}^\epsilon \exp\left(- t^2(n^2-2n-1) \right) \,\text{d}t \right]^{\nicefrac{1}{2}} \\ 
 & \qquad \qquad  \times \left[ \int_{0}^\epsilon \exp\left( -\frac{(1-t^2)^{(n-1)/2}(m-n)}{\sqrt{2n}}\right ) \, \text{d}t\right]^{\nicefrac{1}{2}} \\ 
 & \qquad \le  \left[\epsilon\right]^{\nicefrac{1}{2}}
 \left[\epsilon \exp\!\left(\! -\frac{(1-\epsilon^2)^{(n-1)/2}(m-n)}{\sqrt{2n}}\right ) \right]^{\nicefrac{1}{2}}  \\
 & \qquad =\epsilon \exp\!\left( \!-\frac{(1-\epsilon^2)^{(n-1)/2}(m-n)}{\sqrt{8n}}\right )\!.
 \end{align*}
Replacing the integral with this bound and using the definition $\epsilon = \cos(\phi)$ yields the result.

\end{proof}

\balance

\bibliographystyle{IEEEtran} 
\bibliography{cites}


\begin{IEEEbiographynophoto}{Tom Goldstein} is an Assistant Professor at the University of Maryland in the Department of Computer Science. Before joining the faculty at UMD, Tom completed his PhD at UCLA, and held research positions at Stanford University and Rice University. Tom's research focuses on efficient, low complexity optimization routines. His work ranges from large-scale computing on distributed architectures to inexpensive power-aware algorithms for small-scale embedded systems. Applications of his work include scalable machine learning, computer vision, and signal processing methods for wireless communications. Tom has been the recipient of several awards, including SIAM's DiPrima Prize, and a Sloan Fellowship.
\end{IEEEbiographynophoto}
%


\begin{IEEEbiographynophoto}{Christoph Studer} (S'06--M'10--SM'14) received his Ph.D.\ degree in Electrical Engineering from ETH Zurich in 2009. In 2005, he was a Visiting Researcher with the Smart Antennas Research Group at Stanford University. From 2006 to 2009, he was a Research Assistant in both the Integrated Systems Laboratory and the Communication Technology Laboratory (CTL) at ETH Zurich. From 2009 to 2012, Dr. Studer was a Postdoctoral Researcher at CTL, ETH Zurich, and the Digital Signal Processing Group at Rice University. In 2013, he has held the position of Research Scientist at Rice University. Since 2014, Dr. Studer is an Assistant Professor at Cornell University and an Adjunct Assistant Professor at Rice University. 

Dr. Studer's research interests include signal and information processing as well as the design of digital very large-scale integration (VLSI) circuits. His current research areas include applications in wireless communications, nonlinear signal processing, optimization, and machine learning. 

\vspace{0.08cm}

Dr. Studer received an ETH Medal for his M.S.\ thesis in 2006 and for his Ph.D.\ thesis in 2009. He received a two-year Swiss National Science Foundation fellowship for Advanced Researchers in 2011 and a US National Science Foundation CAREER Award in 2017. In 2016, Dr. Studer won a Michael Tien '72 Excellence in Teaching Award from the College of Engineering, Cornell University. He shared the Swisscom/ICTnet Innovations Award in both 2010 and 2013. Dr. Studer was the winner of the Student Paper Contest of the 2007 Asilomar Conf. on Signals, Systems, and Computers, received a Best Student Paper Award of the 2008 IEEE Int. Symp. on Circuits and Systems (ISCAS), and shared the best Live Demonstration Award at the IEEE ISCAS in 2013.  
\end{IEEEbiographynophoto}


\end{document}